\documentclass[journal]{IEEEtran}

\usepackage{xcolor}

\usepackage{graphicx}

\usepackage{tikz,circuitikz}
\usepackage{subcaption}
\usetikzlibrary{shapes.geometric, arrows.meta, positioning}

\usepackage{pgfplots}
\pgfplotsset{compat=1.18}

\usepackage{cite}

\ifCLASSINFOpdf
\else
\fi

\usepackage{amsmath,amsthm,amssymb}

\usepackage{algorithm}
\usepackage{algpseudocode}

\usepackage{booktabs}

\newtheorem{theorem}{Theorem}
\newtheorem{lemma}[theorem]{Lemma}
\newtheorem{remark}{Remark}

\usepackage{url}%
\usepgfplotslibrary{groupplots}
\newtheorem{proposition}[theorem]{Proposition}%
\newcommand{\ccomp}{c_\mathrm{comp}}%
\newcommand{\csys}{c_\mathrm{sys}}%
\newcommand{\cset}{c_\mathrm{set}}%
\newcommand{\actRep}{\mathrm{Rep}}%
\newcommand{\actCont}{\mathrm{Cont}}%
\newcommand{\repcost}{C}%
\newcommand{\batt}[1]{%
  \begin{tikzpicture}[baseline, every node/.style={}]
    \pgfmathsetmacro{\battfrac}{#1}
    \pgfmathtruncatemacro{\battpct}{round(100*#1)}
    \ifdim\battfrac pt<0.33pt
      \definecolor{battcol}{rgb}{0.85,0.2,0.2}%
    \else\ifdim\battfrac pt<0.66pt
      \definecolor{battcol}{rgb}{0.95,0.8,0.1}%
    \else
      \definecolor{battcol}{rgb}{0.3,0.7,0.2}%
    \fi\fi
    \draw[line width=0.6pt] (0,0) rectangle (1.4,0.5);
    \fill[black] (1.4,0.15) rectangle (1.48,0.35); %
    \fill[battcol, fill opacity=0.5] (0.04,0.04) rectangle (0.04+1.32*\battfrac,0.46);
    \node[black, font=\small, anchor=center, inner sep=0pt] at (0.7,0.25) {\battpct\%};
  \end{tikzpicture}%
}

\begin{document}
\title{Per-State Signatures and Mean Time-to-Failure: An Algorithm for Systems With Simultaneous Failures}

\author{Guido~Lagos
        and
        Sim\'on \'Avila
\thanks{This work was supported by the grants ANID/\allowbreak FONDECYT/\allowbreak Iniciaci\'on~11230256 and ANID/\allowbreak FONDECYT/\allowbreak Regular~1231207. (Corresponding author: Guido Lagos.)}%

\thanks{
    G.~Lagos and S.~\'Avila are with Universidad Adolfo Ib\'a\~nez, Vi\~na del Mar, Chile (e-mail: guido.lagos@uai.cl; simoavila@alumnos.uai.cl).
    }%
}

\maketitle

\begin{abstract}
We present an algorithm that computes the structural signature and the mean time-to-failure (MTTF) of a monotone binary system at every one of its states, rather than only at the state in which all components work.
It makes no assumption on the structure function and evaluates it exactly once per state, as many times as needed for the signature of the complete system alone, so computing the MTTF everywhere costs only a factor polynomial in the number of components over computing it once.
As a by-product, it identifies all minimal path sets.
We assume exchangeable component failure times whose survivor count is a continuous-time Markov chain.
This setting admits simultaneous failures, as in the L\'evy-frailty Marshall-Olkin, binomial failure rate and lethal-shock models, where the simplifications available for single failures do not apply.
We use the per-state MTTF to define threshold repair policies.
We show that the per-state signature is a strictly finer decision criterion than the MTTF: signature dominance at a common failure count implies MTTF dominance, but not conversely.
Under the L\'evy-frailty Marshall-Olkin model the MTTF is monotone in the state, which yields the long-run cost rate of every threshold policy in a single pass.
Experiments on 22 benchmark networks with up to 30 components and 31 failure models show that the best MTTF-threshold policy has a median excess cost of 0.36\% over the optimal state-based policy, and that it is never worse than the best policy that repairs after a fixed number of failures.
\end{abstract}


\IEEEpeerreviewmaketitle

\section{Introduction}\label{sec:intro}

\IEEEPARstart{R}{eliability} indices of a multi-component system are almost always reported for the system with every component working.
Maintenance decisions, however, are taken at partially failed states: an operator observes which components have failed, and must decide whether to keep operating or to intervene.
Two states with the same number of failed components can differ sharply in how much useful life remains, depending on \emph{which} components failed, and a decision rule that ignores this distinction discards information that is available for free.
Acting on it requires the reliability index not at one state, but at all of them.

In this work we provide an algorithm that computes the structural signature and the mean time-to-failure (MTTF) at each of the $2^n$ states of a monotone binary system with $n$ components, at essentially the cost of computing them once.

\subsection{Setting and notation}
We consider a system with $n>1$ components, each of which is either \emph{working} or \emph{failed}.
We represent a configuration of the system by a state $\mathbf x$ in $\{0,1\}^n$, where $x_i=1$ if component $i$ works and $x_i=0$ if it has failed, and we write $|\mathbf x| := \sum_{i=1}^n x_i$ for the number of working components. A deterministic and known \emph{structure function} $\Phi : \{0,1\}^n \to \{0,1\}$ gives the status of the system, with $\Phi(\mathbf x)=1$ when the system works and $0$ otherwise.
The system is \emph{monotone}, meaning that $\Phi(\mathbf x^1) \leq \Phi(\mathbf x^2)$ whenever $\mathbf x^1 \leq \mathbf x^2$ componentwise: a failed system cannot be made to work by breaking down working components, and a working system cannot be made to fail by repairing failed ones.
We also assume the system is non-trivial, i.e., that $\Phi$ is neither identically $0$ nor identically $1$; by monotonicity this is equivalent to $\Phi(\mathbf 0)=0$ and $\Phi(\mathbf 1)=1$, where $\mathbf 0 = (0,\ldots,0)^\mathsf T$ and $\mathbf 1 = (1,\ldots,1)^\mathsf T$.
We do not require every component to be relevant, so semi-coherent systems are covered.

The \emph{structural signature} of the system~\cite{samaniego1985closure,%
samaniego2007system,navarro2021introduction} is the vector $\mathbf s$ in $\{0,1\}^n$ given by
\begin{align}\label{def:s}
    s_k := \bar S_{k-1} - \bar S_k \qquad \text{for } k=1,\ldots, n
\end{align}
with $\bar{\mathbf S}$ in $\{0,1\}^{n+1}$ the \emph{tail signature} vector~\cite{gertsbakh2011signatures}\footnote{A related concept is the sequence $\bar S_{n-l}$ for $l=1,\ldots,n$, where $\bar S_k$ is as defined in~\eqref{def:bar S}, first defined in~\cite{boland2001signatures} and then called \emph{survival signature} in~\cite{coolen2012generalizing}.}
given by
\begin{align}\label{def:bar S}
\bar S_k := \frac {W_k} {W_k+F_k} = \frac{W_k}{\binom n k} \qquad \text{for } k=0,\ldots, n,
\end{align}
where $W_k$ and $F_k$ count the states with $k$ failed components in which the system works, respectively fails.
That is,
\begin{align}
    W_k := & \sum_{\mathbf x \, : \, |\mathbf x|=n-k} \Phi(\mathbf x) \label{def:W}
\end{align}
and $F_k := \sum_{\mathbf x \, : \, |\mathbf x|=n-k} (1-\Phi(\mathbf x))$, with $W_k+F_k = \binom n k$.
Thus, $\bar S_k$ is the proportion, among the $\binom n k$ states $\mathbf x$ with $k$ failed components, in which the system continues to work; this is the concise form of the signature given in~\cite{marichal2011signature}.
By monotonicity and $\Phi(\mathbf 1)=1$ we have $1 = \bar S_0 \geq \bar S_1 \geq \ldots \geq \bar S_n = 0$, so $s_k \geq 0$ for all $k$ and $\sum_{k=1}^n s_k = 1$.

The signature is central to reliability analysis because it separates the combinatorial structure of the system from the probabilistic behavior of its components.
Writing $T_1, \ldots, T_n$ for the components' failure times and $T_{1:n} \leq \ldots \leq T_{n:n}$ for their order statistics, the time $T_{fail}$ at which the system fails satisfies
\begin{align}\label{eq:Samaniego}
    \Pr(T_{fail} > t) = \sum_{k=1}^n s_k \, \Pr(T_{k:n} > t) .
\end{align}
Samaniego~\cite{samaniego1985closure} established~\eqref{eq:Samaniego} for components with independent and identically distributed failure times having an absolutely continuous distribution, and~\cite{navarro2008application} extended it to the case where the joint distribution of $(T_1, \ldots, T_n)$ is \emph{exchangeable}; see Assumption~(A) below.
In particular, \eqref{eq:Samaniego} yields the mean time-to-failure $\mathbb{E} T_{fail} = \sum_{k=1}^n s_k \, \mathbb{E} T_{k:n}$, or equivalently, using~\eqref{def:s} and~\eqref{def:bar S},
\begin{align}\label{eq:E Samaniego}
\mathbb{E} T_{fail} = \sum_{k=0}^{n-1} \frac{W_k}{\binom n k}
\left(\mathbb{E} T_{k+1:n} - \mathbb{E} T_{k:n} \right) ,
\end{align}
with the convention $T_{0:n}:=0$.

\subsection{Computational challenge due to the exponential number of states}

Although $\mathbf s$ and $\bar{\mathbf S}$ have only $n$ entries, obtaining them demands, in principle, the status $\Phi(\mathbf x)$ at all $2^n$ states.
For a system with $n=30$ components this means $2^{30}$ bits, or $128$ megabytes, merely to store the system status; for $n=40$ it means $128$ gigabytes.
Moreover, a single entry can already dominate the count: for $n=40$, computing $\bar S_{\lfloor n/2 \rfloor}$ requires determining how many of the $\binom{40}{20} = 137{,}846{,}528{,}820$ states with $20$ failed components still work.
The difficulty is intrinsic rather than an artifact of this formulation, since exact reliability evaluation is $\#\mathrm{P}$-complete already for networks~\cite{provan1983complexity}.

Computing the signature at \emph{every} state is, at first sight, considerably worse. The subsystem induced by a state $\mathbf x$ ---the system restricted to the components working in $\mathbf x$--- has its own signature, so one would need $2^n$ signature vectors: $\binom n 0$ of dimension $n$, $\binom n 1$ of dimension $n-1$, and so on.
Naively, each of them costs a full enumeration of that subsystem's own states.

\subsection{Contributions}
We show that this pessimism is unwarranted, and that the per-state signatures and MTTFs come at a modest premium over the single-state computation.
Specifically:

\begin{enumerate}
\item We give a recursion (Lemma~\ref{lemma:Wk}) that obtains the counting vectors of a state from those of its immediate sub-states, and an algorithm (Algorithm~\ref{alg}) that sweeps the state space by levels of failure count. Theorem~\ref{theo} shows that the algorithm evaluates the structure function exactly $2^n$ times ---the same number required to compute the signature of the complete system alone--- and that computing the MTTF at all states, instead of at one, costs only a polynomial factor in $n$ in memory and in basic arithmetic operations.
No assumption whatsoever is made on $\Phi$.

\item We work under exchangeable failure times whose survivor-count process is a continuous-time Markov chain.
This admits simultaneous failures of several components, as caused by common-cause shocks, and it is exactly the regime in which the simplifications available for single failures break down.
We delimit the resulting model class precisely, and illustrate it with several families of failure models, including the L\'evy-frailty Marshall-Olkin family and the binomial failure rate and lethal-shock models of probabilistic safety assessment.

\item As a corollary, the same sweep identifies all minimal path sets of the system, at no additional cost in evaluations of $\Phi$.

\item We use the per-state MTTF to define threshold repair policies, and compare them through their long-run expected cost rate.
We further show (Theorem~\ref{theo:dominance}) that signature dominance at a common failure count implies MTTF dominance, and that the MTTF is a single linear projection of the per-state signature; consequently, signature-threshold policies strictly refine MTTF-threshold policies.
Under the L\'evy-frailty Marshall-Olkin distribution the MTTF is monotone in the state (Lemma~\ref{lemma:monotone}), and we show that this yields the long-run cost of \emph{all} MTTF-threshold policies in a single pass over the states, grouped into classes with the same counting vector (Proposition~\ref{prop:all thresholds}). All costs include a setup cost per repair intervention. We also compute the optimal state-based policy, which bounds what any policy based on the per-state signature can achieve.

\item We report computational experiments on 22 benchmark networks with up to 30 components, 31 failure models calibrated to a common pairwise beta factor, and 45 cost settings (Section~\ref{sec:experiments}). The best MTTF-threshold policy has a median excess cost of $0.36\%$ over the optimal state-based policy, and it is never worse than the best $r$-out-of-$n$ policy, over which it saves a median of $4.2\%$; its cost curve, however, is neither convex nor unimodal in the threshold, i.e., it can increase and then decrease again, so the threshold must be tuned over all breakpoints.
\end{enumerate}

\subsection{Organization}
Section~\ref{sec:related} reviews related work. Section~\ref{sec:subsystems} states our
assumptions on the failure-time distribution and formulates the per-state MTTF.
Section~\ref{sec:procedure} presents the recursion, the algorithm and its complexity
analysis. Section~\ref{sec:policies} develops the MTTF-threshold repair policies and their long-run
cost, and the optimal state-based policy. Section~\ref{sec:dominance} establishes the dominance and expressiveness results.
Section~\ref{sec:experiments} reports the computational experiments, and Section~\ref{sec:conclusions} concludes. The proofs of all results are in Appendix~\ref{app:proofs}.

\section{Related work}\label{sec:related}

\paragraph{Signatures and their computation}
Since its introduction by Samaniego~\cite{samaniego1985closure}, the signature has become a standard tool for comparing and evaluating coherent systems~\cite{samaniego2007system,navarro2021introduction}: stochastic orders between systems follow from orders between their signatures~\cite{kochar1999signature}, and the decomposition~\eqref{eq:Samaniego} extends from iid to exchangeable lifetimes~\cite{navarro2008application}, the setting we adopt. Marichal~\emph{et~al.}~\cite{marichal2011signature} gave the counting form~\eqref{def:bar S}, and Coolen and Coolen-Maturi~\cite{coolen2012generalizing} the survival signature for several component types.
Computing the signature exactly is the practical bottleneck. Existing methods enumerate minimal path or cut sets~\cite{kvassay2017minimal,rebaiaia2012algorithms,yeh2024new}, or use symbolic representations of the structure function such as binary decision diagrams~\cite{kuo2007efficient}, and they target the signature of the \emph{complete} system. We need one signature per state, i.e., $2^n$ of them, and we are not aware of an algorithm that produces these, for an arbitrary structure function, at a cost comparable to a single signature computation.

\paragraph{Reliability of used and partially failed systems}
Mixture representations of the residual lifetime of a used system~\cite{navarro2008mixture}, and of its conditional residual life given the number of failures and the system status~\cite{parvardeh2014conditional}, condition on the \emph{number} of failed components rather than on their identity. They are representations rather than algorithms: obtaining the value at every state from them would repeat the derivation state by state, without exploiting that neighboring states share almost all of their sub-states, which is what Lemma~\ref{lemma:Wk} does.

\paragraph{Exchangeable models with a Markov survivor count}
Ross, Shahshahani and Weiss~\cite{ross1980number} identified the structure on which Assumptions~(A)--(B) of Section~\ref{sec:subsystems} rest; the model of Ross~\cite{ross1984model}, in its exchangeable case, and sequential order statistics with exponential baseline~\cite{kamps1995concept} are single-failure instances. Together, (A) and (B) force a pure-death chain on the number of working components, which excludes models based on Markovian arrival processes~\cite{garcia2018modeling,yera2021bivariate}, semi-Markov multi-state models~\cite{barbu2004discrete} and exchangeable dependent-renewal constructions~\cite{coen2019modelling}; among the models they admit, the L\'evy-frailty Marshall-Olkin family is the richest tractable one with simultaneous failures.

\paragraph{Simultaneous and common-cause failures}
The multivariate exponential distribution of Marshall and Olkin~\cite{marshall1967multivariate} is the canonical model of common shocks; its specializations in probabilistic safety assessment are the binomial failure rate~\cite{vesely1977estimating,atwood1986binomial} and lethal-shock~\cite{atwood1986binomial} models. The L\'evy-frailty Marshall-Olkin family~\cite{mai2009levy} is its exchangeable subclass driven by a L\'evy subordinator, whose order statistics and asymptotics were analyzed in~\cite{barrera2020limit,lagos2024limiting}, and repair policies and decompositions in~\cite{lagos2026simple}. We note in Section~\ref{subsec:models} that the binomial failure rate and lethal-shock models are of L\'evy-frailty type, so this family covers the standard common-cause models.

\paragraph{Maintenance driven by reliability indices and by signatures}
Signature-based maintenance is surveyed in~\cite{asadi2023overview}. Finkelstein~\emph{et~al.}~\cite{finkelstein2015time} proposed ``time-free'' policies triggered by the number of failed components; in these and related policies the trigger is a scalar index, and the signature enters only through the expected costs and cycle lengths. Closest to our setting, Feng and Yun~\cite{feng2024mean} replace all failed components once the mean residual life of the system drops below a tuned threshold; Section~\ref{sec:policies} discusses why their computation is not viable beyond very small systems.

\section{MTTF at all states}\label{sec:subsystems}

We focus on computing the signature of the subsystem obtained at each state $\mathbf x$ in $\{0,1\}^n$, i.e., when restricting the system to the components that are working in state $\mathbf x$.
Our main motivation for this, is the computation of the mean time-to-failure (MTTF) when the system is in any given state $\mathbf x$, and also derivation of repair policies based on these times.

We will do this in the following setting for the joint probability distribution of the vector $(T_1, \ldots, T_n)$ of failure times of components.
\subsubsection*{Assumptions}
\begin{description}

    \item[(A)] The vector $(T_1, \ldots, T_n)$ has exchangeable (EXC) components, meaning that for any two permutation $\sigma^1$ and $\sigma^2$ of \{1,\ldots, n\}, the random vectors $(T_{\sigma^1(1)}, \ldots, T_{\sigma^1(n)})$ and $(T_{\sigma^2(1)}, \ldots, T_{\sigma^2(n)})$ have the same distribution.

    \item[(B)] The process $N(t) = \sum_{k=1}^n 1_{\{T_k > t\}}$ of number of working components at all times $t\geq 0$, is distributed as a continuous-time Markov chain. 

    \item[(C)] For all states $\mathbf x$ in $\{0,1\}^n$, $\mathbf x \neq \mathbf 0$, we can compute the expected time $\mathbb{E} T^{\mathbf x}_{k:|\mathbf x|}$ of $k$-th distinct failure of components, when the system starts from state $\mathbf x$, for $k=1, \ldots, |\mathbf x|$.
    
\end{description}
We remark that in this case, due to Assumption~(A), when the system has $w_1>0$ working components, the time until the next failure is distributed exponentially with rate $\nu_{w_1}>0$ and the next number of working components is $w_2$, $0 \leq w_2 < w_1$, with probability $P_{w_1,w_2} \geq 0$. Also, which $w_1-w_2$ components fail, between the $w_1$ working ones, is taken uniformly at random.
Additionally, the times $T^{\mathbf x}_{k:|\mathbf x|}$ depend only on $|\mathbf x|$, due to the exchangeability of components stated in Assumption~(A), so onwards we write $\mathbb{E} T_{k:|\mathbf x|}$.

Section~\ref{subsec:models} presents several examples of failure models
satisfying Assumptions~(A)--(C), including models that allow simultaneous
component failures.
It is for the latter class that our method is of most
interest, since the simplifications available in the case of
single failures do not apply.

In this setting, the MTTF $m(\mathbf x)$ when the system arrives to state $\mathbf x$ in $\{0,1\}^n$, is $m(\mathbf x) = 0$ if $\Phi(\mathbf x)=0$, and otherwise,
\begin{align}\label{eq:mttf}
    m(\mathbf x) = \frac{1}{\nu_{|\mathbf x|}} + \sum_{\mathbf x' \, : \, \mathbf x' \leq \mathbf x} P_{|\mathbf x|, |\mathbf x'|} \frac{1}{\binom{|\mathbf x|}{|\mathbf x|-|\mathbf x'|}} \ m(\mathbf x').
\end{align}

The computation of the mean times $m(\mathbf x)$ at all states $\mathbf x$ in $\{0,1\}^n$ in this setting can be computationally challenging, as the system defined by~\eqref{eq:E Samaniego} has $2^n$ unknown variables ---one for each $m(\mathbf x)$--- over $2^n$ equations.
Only in simplified settings can these times $m(\mathbf x)$ be more easily computed, such as when the components failure times are iid exponentially distributed.
In that case, $P_{w,w-1}=1$, and the system can be easily solved due to the sparsity of the matrix $P$.
Another tractable example is in the so-called $k$-out-of-$n$:F case, where the system fails when there are $k$ or more failed components. In that case, the system of equations is not sparse, but the variables $m(\mathbf x)$ can be reduced to only $k$ unknown variables.
On the other hand, a case where no such simplification is possible is with the L\'evy-frailty Marshall-Olkin distribution, see~\cite{lagos2026simple}.
This is a noteworthy case, as it allows to model simultaneous failures of components, is easy to simulate and allows for approximations; see~\cite{lagos2024limiting}.

Together, Assumptions~(A) and~(B) are considerably more restrictive than either alone: they force a pure-death chain on the number of working components, and exclude several standard classes of dependent failure models (Section~\ref{sec:related}).

\section{Efficient computation through signatures}\label{sec:procedure}

The approach that we take to obtain the MTTF at all states is to use the decomposition~\eqref{eq:E Samaniego} that uses the signature of the system. %
Nonetheless, the difficulty now is that we need the signature vector \emph{for each state} of the system. So in principle, now we would need to compute $2^n$ signature vectors: $\binom n 0$ of dimension $n$ (i.e., for the state $\mathbf 1$ with no failures), $\binom n 1$ of dimension $n-1$ (i.e., for the $n$ states with 1 failure), and so on.

To solve this, our main result, Theorem~\ref{theo} below, states that the procedure presented in Algorithm~\ref{alg} below can efficiently compute the signature and mean time-to-failure at each state.

We will use the following nomenclature for the vectors $\mathbf s$, $\bar{\mathbf S}$ and $\mathbf W$, defined in~\eqref{def:s}, \eqref{def:bar S} and~\eqref{def:W}, respectively, when restricted to a state $\mathbf x$. %
Define for all $\mathbf x$ in $\{0,1\}^n$, $W^{\mathbf x}_k := \sum_{\mathbf x' \leq \mathbf x \ : \ |\mathbf x'|=|\mathbf x|-k} \Phi(\mathbf x')$, i.e., the number of sub-states of $\mathbf x$ where the system is working after shutting-down $k$ working components of $\mathbf x$, for $k=0, \ldots, |\mathbf x|$.
See Figure~\ref{fig:lemma} for an illustration of the counting vectors $\mathbf W^{\mathbf x}$ at all states $\mathbf x$ of a system with $n=3$ components.
Also, for $\mathbf x \neq \mathbf 0$, define $\bar S^{\mathbf x}_k := W^{\mathbf x}_k / \binom{|\mathbf x|}{k}$ and $s^{\mathbf x}_k := \bar S^{\mathbf x}_{k-1} - \bar S^{\mathbf x}_k$ for $k=1, \ldots, |\mathbf x|$, as the tail signature and structural signature localized at state $\mathbf x$, respectively.

With this, we will compute the MTTF $m(\mathbf x)$ at every state $\mathbf x$ with
\begin{align}\label{eq:E Samaniego sub}
m(\mathbf x) = \sum_{k=0}^{|\mathbf x|-1} \frac{W^{\mathbf x}_k}{\binom{|\mathbf x|}{k}} \left( \mathbb{E} T_{k+1:|\mathbf x|} - \mathbb{E} T_{k:|\mathbf x|} \right) 
\end{align}
that comes from decomposition~\eqref{eq:E Samaniego}.
Note that formula \eqref{eq:E Samaniego sub} is also valid when $\Phi(\mathbf x)=0$, since then $m(\mathbf x)=0$ and $\mathbf W^{\mathbf x} = \mathbf 0$.

The following lemma shows that we can compute recursively the signature vectors $\mathbf s^{\mathbf x}$ at all states $\mathbf x$, noting first that $\mathbf s^{\mathbf x}$ is essentially based on the counting vector $\mathbf W^{\mathbf x}$: if at all states $\mathbf x^k$ having $k$ failures we have the counting vectors $\mathbf W^{\mathbf x^k}$, we can compute the vectors $\mathbf W^{\mathbf x^{k-1}}$ at all states $\mathbf x^{k-1}$ with $k-1$ failures.

\begin{lemma}\label{lemma:Wk}
    For all state $\mathbf x$ in $\{0,1\}^n$, $W^{\mathbf x}_0 = \Phi(\mathbf x)$, and if $\mathbf x \neq \mathbf 0$ it holds that
    \begin{align}\label{eq:Wk}
        W^{\mathbf x}_k = \frac{1}{k} \sum_{\mathbf x' \in \mathrm{Sons}(\mathbf x)} W^{\mathbf x'}_{k-1} \qquad \text{for } k=1,\ldots, |\mathbf x|
    \end{align}
    where $\mathrm{Sons}(\mathbf x) = \{ \mathbf x' \in \{0,1\}^n \ : \ \mathbf x' \leq \mathbf x \text{ and } |\mathbf x'|=|\mathbf x|-1\}$ is the set of states obtained by shutting down one working component of $\mathbf x$.
\end{lemma}

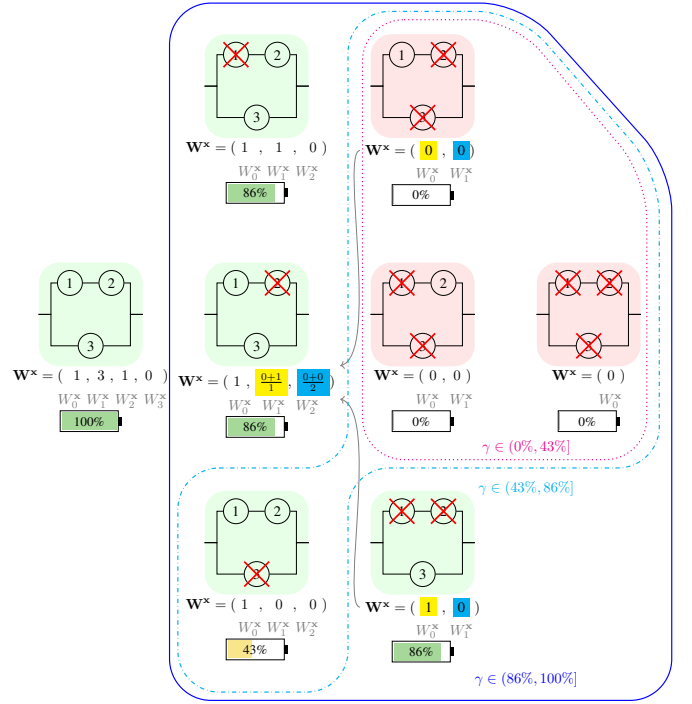
\begin{figure}[h]%
\centering
\resizebox{1\linewidth}{!}{%
\begin{tikzpicture}[
    component/.style={draw, circle, inner sep=1pt, minimum size=0.6cm},
    blank_system/.style={rounded corners=5mm, minimum width=2.5cm, minimum height=2.5cm, align=center},
    working_system/.style={fill=green!10, rounded corners=5mm, minimum width=2.5cm, minimum height=2.5cm, align=center},
    failed_system/.style={fill=red!10, rounded corners=5mm, minimum width=2.5cm, minimum height=2.5cm, align=center},
    cross out/.pic={
        \draw[red, very thick] (-0.3, -0.3) -- (0.3, 0.3);
        \draw[red, very thick] (-0.3, 0.3) -- (0.3, -0.3);
    }
]

\node[failed_system] (sys000) at (12, -5.5) {};
\node[component] (c4_1) at (sys000.north) [xshift=-0.5cm, yshift=-0.5cm] {1};
\node[component] (c4_2) at (sys000.north) [xshift= 0.5cm, yshift=-0.5cm] {2};
\node[component] (c4_3) at (sys000.south) [yshift= 0.5cm] {3};
\draw (c4_1.east) -- (c4_2.west);
\draw (sys000.east) -| +(-.3,0) |- (c4_2.east);
\draw (sys000.east) -| +(-.3,0) |- (c4_3.east);
\draw (sys000.west) -| +( .3,0) |- (c4_3.west);
\draw (sys000.west) -| +( .3,0) |- (c4_1.west);
\pic at (c4_1) {cross out};
\pic at (c4_2) {cross out};
\pic at (c4_3) {cross out};

\node[failed_system] (sys100) at (8, -0) {};
\node[component] (c9_1) at (sys100.north) [xshift=-0.5cm, yshift=-0.5cm] {1};
\node[component] (c9_2) at (sys100.north) [xshift= 0.5cm, yshift=-0.5cm] {2};
\node[component] (c9_3) at (sys100.south) [yshift= 0.5cm] {3};
\draw (c9_1.east) -- (c9_2.west);
\draw (sys100.east) -| +(-.3,0) |- (c9_2.east);
\draw (sys100.east) -| +(-.3,0) |- (c9_3.east);
\draw (sys100.west) -| +( .3,0) |- (c9_3.west);
\draw (sys100.west) -| +( .3,0) |- (c9_1.west);
\pic at (c9_2) {cross out};
\pic at (c9_3) {cross out};

\node[failed_system] (sys010) at (8, -5.5) {};
\node[component] (c6_1) at (sys010.north) [xshift=-0.5cm, yshift=-0.5cm] {1};
\node[component] (c6_2) at (sys010.north) [xshift= 0.5cm, yshift=-0.5cm] {2};
\node[component] (c6_3) at (sys010.south) [yshift= 0.5cm] {3};
\draw (c6_1.east) -- (c6_2.west);
\draw (sys010.east) -| +(-.3,0) |- (c6_2.east);
\draw (sys010.east) -| +(-.3,0) |- (c6_3.east);
\draw (sys010.west) -| +( .3,0) |- (c6_3.west);
\draw (sys010.west) -| +( .3,0) |- (c6_1.west);
\pic at (c6_1) {cross out};
\pic at (c6_3) {cross out};

\node[working_system] (sys001) at (8, -11) {};
\node[component] (c7_1) at (sys001.north) [xshift=-0.5cm, yshift=-0.5cm] {1};
\node[component] (c7_2) at (sys001.north) [xshift= 0.5cm, yshift=-0.5cm] {2};
\node[component] (c7_3) at (sys001.south) [yshift= 0.5cm] {3};
\draw (c7_1.east) -- (c7_2.west);
\draw (sys001.east) -| +(-.3,0) |- (c7_2.east);
\draw (sys001.east) -| +(-.3,0) |- (c7_3.east);
\draw (sys001.west) -| +( .3,0) |- (c7_3.west);
\draw (sys001.west) -| +( .3,0) |- (c7_1.west);
\pic at (c7_1) {cross out};
\pic at (c7_2) {cross out};

\node[working_system] (sys011) at (4, 0) {};
\node[component] (c2_1) at (sys011.north) [xshift=-0.5cm, yshift=-0.5cm] {1};
\node[component] (c2_2) at (sys011.north) [xshift= 0.5cm, yshift=-0.5cm] {2};
\node[component] (c2_3) at (sys011.south) [yshift= 0.5cm] {3};
\draw (c2_1.east) -- (c2_2.west);
\draw (sys011.east) -| +(-.3,0) |- (c2_2.east);
\draw (sys011.east) -| +(-.3,0) |- (c2_3.east);
\draw (sys011.west) -| +( .3,0) |- (c2_3.west);
\draw (sys011.west) -| +( .3,0) |- (c2_1.west);
\pic at (c2_1) {cross out};

\node[working_system] (sys101) at (4, -5.5) {};
\node[component] (c5_1) at (sys101.north) [xshift=-0.5cm, yshift=-0.5cm] {1};
\node[component] (c5_2) at (sys101.north) [xshift= 0.5cm, yshift=-0.5cm] {2};
\node[component] (c5_3) at (sys101.south) [yshift= 0.5cm] {3};
\draw (c5_1.east) -- (c5_2.west);
\draw (sys101.east) -| +(-.3,0) |- (c5_2.east);
\draw (sys101.east) -| +(-.3,0) |- (c5_3.east);
\draw (sys101.west) -| +( .3,0) |- (c5_3.west);
\draw (sys101.west) -| +( .3,0) |- (c5_1.west);
\pic at (c5_2) {cross out};

\node[working_system] (sys110) at (4, -11) {};
\node[component] (c8_1) at (sys110.north) [xshift=-0.5cm, yshift=-0.5cm] {1};
\node[component] (c8_2) at (sys110.north) [xshift= 0.5cm, yshift=-0.5cm] {2};
\node[component] (c8_3) at (sys110.south) [yshift= 0.5cm] {3};
\draw (c8_1.east) -- (c8_2.west);
\draw (sys110.east) -| +(-.3,0) |- (c8_2.east);
\draw (sys110.east) -| +(-.3,0) |- (c8_3.east);
\draw (sys110.west) -| +( .3,0) |- (c8_3.west);
\draw (sys110.west) -| +( .3,0) |- (c8_1.west);
\pic at (c8_3) {cross out};

\node[working_system] (sys111) at (0, -5.5) {};
\node[component] (c1_1) at (sys111.north) [xshift=-0.5cm, yshift=-0.5cm] {1};
\node[component] (c1_2) at (sys111.north) [xshift= 0.5cm, yshift=-0.5cm] {2};
\node[component] (c1_3) at (sys111.south) [yshift= 0.5cm] {3};
\draw (c1_1.east) -- (c1_2.west);
\draw (sys111.east) -| +(-.3,0) |- (c1_2.east);
\draw (sys111.east) -| +(-.3,0) |- (c1_3.east);
\draw (sys111.west) -| +(.3,0) |- (c1_3.west);
\draw (sys111.west) -| +(.3,0) |- (c1_1.west);

\node (W000) [below=1.2cm] at (sys000) {$\mathbf W^{\mathbf x} = (\ 0\ )$};
\node (L000) [below=.3cm] at (W000) {\color{gray} \qquad\quad $W^{\mathbf x}_0$};

\node (W100) [below=1.2cm] at (sys100) {$\mathbf W^{\mathbf x} = (\ \colorbox{yellow}{0} \ , \ \colorbox{cyan}{0}) \ $};
\node (L100) [below=.3cm] at (W100) {\color{gray} \qquad\quad $W^{\mathbf x}_0$ \ $W^{\mathbf x}_1$};

\node (W010) [below=1.2cm] at (sys010) {$\mathbf W^{\mathbf x} = ( \ 0 \ , \  0 \ )$};
\node (L010) [below=.3cm] at (W010) {\color{gray} \qquad\quad $W^{\mathbf x}_0$ \ $W^{\mathbf x}_1$};

\node (W001) [below=1.2cm] at (sys001) {$\mathbf W^{\mathbf x} = ( \ \colorbox{yellow}{1} \ , \  \colorbox{cyan}{0} \ )$};
\node (L001) [below=.3cm] at (W001) {\color{gray} \qquad\quad $W^{\mathbf x}_0$ \ $W^{\mathbf x}_1$};

\node (W011) [below=1.2cm] at (sys011) {$\mathbf W^{\mathbf x} = ( \ 1 \ \ , \ \ 1 \ \ , \ \ 0 \ )$};
\node (L011) [below=.3cm] at (W011) {\color{gray} \qquad\quad $W^{\mathbf x}_0$ $W^{\mathbf x}_1$ $W^{\mathbf x}_2$};

\node (W101) [below=1.2cm] at (sys101) {$\mathbf W^{\mathbf x} = (\ 1\ , \, \colorbox{yellow}{$\frac{0+1}{1}$}, \, \colorbox{cyan}{$\frac{0+0}{2}$})$};
\node (L101) [below=.3cm] at (W101) {\color{gray} \qquad\ $W^{\mathbf x}_0$ \ $W^{\mathbf x}_1$ \  $W^{\mathbf x}_2$};

\node (W110) [below=1.2cm] at (sys110) {$\mathbf W^{\mathbf x} = ( \ 1 \ \ , \ \ 0 \ \ , \ \ 0 \ )$};
\node (L110) [below=.3cm] at (W110) {\color{gray} \qquad\quad $W^{\mathbf x}_0$ $W^{\mathbf x}_1$ $W^{\mathbf x}_2$};

\node (W111) [below=1.2cm] at (sys111) {$\mathbf W^{\mathbf x} = ( \ \ 1 \ , \ 3 \ , \ 1 \ , \ 0 \ \ )$};
\node (L111) [below=.3cm] at (W111) {\color{gray} \qquad\quad $W^{\mathbf x}_0$ $W^{\mathbf x}_1$ $W^{\mathbf x}_2$ $W^{\mathbf x}_3$};

\draw[->, gray] (W100) .. controls +(-2,0) and +(3, 0.6) .. (W101);
\draw[->, gray] (W001) .. controls +(-2,0) and +(3,-0.6) .. (W101);

\node (B111) [below=.1cm] at (L111) {\batt{1.0}};
\node (B011) [below=.1cm] at (L011) {\batt{0.857}};
\node (B101) [below=.1cm] at (L101) {\batt{0.857}};
\node (B110) [below=.1cm] at (L110) {\batt{0.429}};
\node (B100) [below=.1cm] at (L100) {\batt{0.0}};
\node (B010) [below=.1cm] at (L010) {\batt{0.0}};
\node (B001) [below=.1cm] at (L001) {\batt{0.857}};
\node (B000) [below=.1cm] at (L000) {\batt{0.0}};

\draw [magenta, dotted, thick, rounded corners=1cm] (6.55, 1.6) -- (10.2, 1.6) -- (13.4,-2.2) -- (13.4,-9.0) -- (6.55, -9.0) -- cycle;
\node[magenta, anchor=south] at (10.5, -9.0) {$\gamma \in (0\%, 43\%]$};

\draw [cyan, dashdotted, thick, rounded corners=1cm] (2.1, -9.2) -- (6.2, -9.2) -- (6.2, 1.8) -- (10.3, 1.8) -- (13.7, -2.1) -- (13.7, -9.2) -- (6.2, -9.2) -- (6.2, -14.6) -- (2.1, -14.6) -- cycle;
\node[cyan, anchor=north] at (10.5, -9.4) {$\gamma \in (43\%, 86\%]$};

\draw [blue, thick, rounded corners=1cm] (1.9, 2.0) -- (10.4, 2.0) -- (14., -2.0) -- (14., -14.8) -- (1.9, -14.8) -- cycle;
\node[blue, anchor=north] at (10.5, -14) {$\gamma \in (86\%, 100\%]$};
\end{tikzpicture}
}
\caption{
All states $\mathbf x$ of a system with $n=3$ components that works iff there is a path of working components from left to right, i.e., $\Phi(\mathbf x) = \max\{x_3, \min\{x_1, x_2\}\}$; crosses mark failed components, and green (red) backgrounds working (failed) states.
Below each state $\mathbf x$ we show its counting vector $\mathbf W^{\mathbf x}$; for instance, $\mathbf W^{\mathbf 1} = (1,3,1,0)$ gives the tail signature $\bar{\mathbf S} = (1, 1, 1/3, 0)$ and the signature $\mathbf s = (0, 2/3, 1/3)$ of the system.
First, $W^{\mathbf x}_0 = \Phi(\mathbf x)$ for every state $\mathbf x$.
Second, by Lemma~\ref{lemma:Wk}, for instance for the state $\mathbf x = (1,0,1)$, $\mathbf W^{(1,0,1)}$ is computed from $\Phi(1,0,1)$ and the counting vectors of its two sons $(1,0,0)$ and $(0,0,1)$:
by~\eqref{eq:Wk}, for $k=1$ (yellow background) and $k=2$ (blue background), $W^{(1,0,1)}_k = \big( W^{(1,0,0)}_{k-1} + W^{(0,0,1)}_{k-1} \big) / k$.
The batteries show the remaining life $m(\mathbf x)/m(\mathbf 1)$ for iid exponentially distributed failure times with unit mean, with $m(\mathbf 1) = 7/6$, and the three boundaries enclose the states at which the MTTF-threshold policies $\pi^\gamma$ of Section~\ref{subsec:threshold} repair, which are the only three distinct ones. With $\ccomp=1/30$, $\csys=1$ and $\cset=0$, their long-term mean costs are $1/10$, $11/15$ and $97/105$ for $\gamma$ in $(86\%,100\%]$, $(43\%,86\%]$ and $(0\%,43\%]$, respectively, so repairing at the first failure is the cheapest.
}
\label{fig:lemma}
\end{figure}

The proof, in Appendix~\ref{app:proofs}, uses that each sub-state of $\mathbf x$ with $k$ fewer working components is a sub-state of exactly $k$ sons of $\mathbf x$.

\begin{algorithm}
\caption{Computation of mean time-to-failure at all states}\label{alg}
\begin{algorithmic}
\Require number $n$ of components, and system function $\Phi$.

\For{$k = n, n-1, \ldots, 0$}
\If{$k=n$}
    \State $W^{\mathbf 0}_0 \gets 0$ and $m(\mathbf 0) \gets 0$
    \State Store in memory $W^{\mathbf 0}$ and $m(\mathbf 0)$.
\Else
    \State For all state $\mathbf x^{k}$ with $k$ failures, compute $W^{\mathbf x^k}_0$, \ldots, $W^{\mathbf x^k}_{|\mathbf x^k|}$ using formula~\eqref{eq:Wk}, and compute $m(\mathbf x^k)$ using~\eqref{eq:E Samaniego sub}.
    \State Store in memory the MTTF $m(\mathbf x^k)$ and the vector $\mathbf W^{\mathbf x^k}$, for all states $\mathbf x^{k}$ with $k$ failures.
    \State Erase from memory the vectors $\mathbf W^{\mathbf x^{k+1}}$ for all states $\mathbf x^{k+1}$ with $k+1$ failures.
\EndIf
\EndFor
\end{algorithmic}
\end{algorithm}

The following theorem states the computational efficiency of Algorithm~\ref{alg} to compute the MTTF at each state.
In short, it states that the algorithm uses the same number of accesses to the structural function $\Phi$ of the system, and that there is only a polynomial factor (over the number of components $n$) to pay, in memory and in simple algebraic operations, in order to compute the MTTF at all states instead of just for the complete system.

\begin{theorem}\label{theo}
Suppose that Assumptions~(A), (B) and~(C) hold. Then the following apply.
\begin{enumerate}
    \item Algorithm~\ref{alg} evaluates the structure function $\Phi$ exactly $2^n$ times, once per state, which is the number of evaluations needed to compute the signature~\eqref{def:s} of the complete system.

    \item Computing the MTTF~\eqref{eq:E Samaniego} of the complete system requires $2^n$ additions of bits and $\mathcal{O}(n)$ other basic algebraic operations, and, processing the states one level at a time, $\mathcal{O}(2^n / \sqrt{n})$ bits of memory.

    \item Computing the MTTF at every state with Algorithm~\ref{alg} requires $2^n$ floats of memory---one for the MTTF of each state---, $\mathcal{O}(n^2 2^n)$ integer additions and memory accesses, $\mathcal{O}(n 2^n)$ other basic algebraic operations, and $\mathcal{O}(\sqrt n \, 2^n)$ integers of memory for the counting vectors.
\end{enumerate}
\end{theorem}

The proof is in Appendix~\ref{app:proofs}.

We remark that, as an interesting corollary, Algorithm~\ref{alg} can also identify the so-called \emph{minimal path-sets} (MPS) of the system $\Phi$.
A set of components, say $A \subseteq \{1,\ldots,n\}$, is a MPS if the state $\mathbf 1_A := (1_{\{ i \in A\}} : i=1,\ldots,n)$ is such that $\Phi(\mathbf 1_A)=1$ and $\Phi(\mathbf x')=0$ at all state $\mathbf x'$ in $\{0,1\}^n$ such that $\mathbf x' \leq \mathbf 1_A$ and $\mathbf 1_A \neq \mathbf x'$.
Indeed, note that by Lemma~\ref{lemma:Wk}, a state $\mathbf x$ defines a MPS (i.e., the set $A^{\mathbf x}:=\{i : x_i=1\}$ is a MPS) iff $\mathbf x$ is working and all its sons are failed, i.e., $W^{\mathbf x}_0=1$ and $W^{\mathbf x}_1=0$.
Consequently, Algorithm~\ref{alg} can identify \emph{all} MPS of the systems by flagging all the states $\mathbf x$ with $W^{\mathbf x}_0=1$ and $W^{\mathbf x}_1=0$.

\section{Application: repair policies based on the MTTF}\label{sec:policies}

An application of our algorithm is to derive repair policies based on the MTTF of the system, as in Feng and Yun~\cite{feng2024mean}.
The latter work replaces all failed components once the mean residual life of the system falls below a threshold, and then tunes the threshold to minimize the long-run expected cost rate.
Their mean residual lives are, however, computed by brute force: a depth-first traversal of all \emph{ordered} sequences of component failures, evaluating the system status at every
node, so that a configuration with $k$ failed components is revisited once for each of the $k!$ orders in which those failures may have occurred and the traversal explores up to
$e \cdot n!$ nodes rather than the $2^n$ states.
Their derivation is moreover tied to a specific family of systems ---planar coverage, or ``lamp'', systems--- and to independent and identically distributed exponential lifetimes, hence to single failures, with experiments reported for only $n=6$ and $9$ components.

We will consider state-based repair policies $\pi: \{0,1\}^n \to \{\actRep, \actCont\}$ where $\actRep$ means \emph{repairing} all failed components and the system, if it failed, and $\actCont$ means \emph{continue} operating as is. We assume that the action is taken as soon as the system arrives to a state $\mathbf x$ in $\{0,1\}^n$ of the system.

For a given policy $\pi$, the mean time-to-repair $m_{\mathrm{rep}}^\pi (\mathbf x)$ when the system starts in state $\mathbf x$ in $\{0,1\}^n$ and operates under the policy $\pi$, can be computed by solving the following linear system of $2^n$ equations (one for each $\mathbf x$) and $2^n$ variables (one for each $m_{\mathrm{rep}}^\pi(\mathbf x)$):
\begin{align}
m_{\mathrm{rep}}^\pi(\mathbf x) = \begin{cases}
\frac{1}{\nu_{|\mathbf x|}} + \sum_{\mathbf x' \, : \, \mathbf x' \leq \mathbf x}  \frac{ P_{|\mathbf x|, |\mathbf x'|} }{ \binom{|\mathbf x|}{|\mathbf x|-|\mathbf x'|} } \ m_{\mathrm{rep}}^\pi (\mathbf x') \\
\hphantom{0 \qquad\qquad\qquad\qquad} \text{if } \pi(\mathbf x)=\actCont \\
0 \qquad\qquad\qquad\qquad \text{if } \pi(\mathbf x)=\actRep
\end{cases} \nonumber \\
\text{for all } \mathbf x \in \{0,1\}^n. 
\label{sys:m rep}
\end{align}
For the costs, let $\ccomp$ be the cost to replace each failed component, $\csys$ the cost to restore a failed system, and $\cset \geq 0$ a fixed cost incurred by every repair intervention, e.g., the mobilization of the repair crew or a production stop. A repair at state $\mathbf x$ then costs
\begin{align}\label{eq:cost setup}
\repcost(\mathbf x) := \cset + (n - |\mathbf x|) \, \ccomp + \mathbf 1_{\{\Phi(\mathbf x) = 0\}} \, \csys .
\end{align}
The setup cost creates economies of scale in grouping replacements: without it, repairing at the first failure is often close to optimal.
The mean cost $c_{\mathrm{rep}}^\pi(\mathbf x)$ until the first repair, when starting at state $\mathbf x$ in $\{0,1\}^n$, is given by the following linear system of $2^n$ equations and $2^n$ variables:
\begin{align}
c_{\mathrm{rep}}^\pi(\mathbf x) = \begin{cases}
\sum_{\mathbf x' \, : \, \mathbf x' \leq \mathbf x} \frac{ P_{|\mathbf x|, |\mathbf x'|} }{ \binom{|\mathbf x|}{|\mathbf x|-|\mathbf x'|} } \ c_{\mathrm{rep}}^\pi (\mathbf x') & \text{if } \pi(\mathbf x)=\actCont \\
\repcost(\mathbf x) & \text{if } \pi(\mathbf x)=\actRep
\end{cases} \nonumber \\
\text{for all } \mathbf x \in \{0,1\}^n.
\label{sys:c rep}
\end{align}
The long-term mean cost of the system operating under a policy $\pi$ that continues at $\mathbf 1$ is
\begin{align}\label{eq:LTMC pi}
    LTMC^\pi = \frac{ c_{\mathrm{rep}}^{\pi}(\mathbf 1) }{ m_{\mathrm{rep}}^{\pi}(\mathbf 1) }.
\end{align}
Indeed, since the system is Markovian and a repair brings it back to the brand new state $\mathbf 1$, the repairs are renewal epochs; the expected cost of a renewal cycle is $c_{\mathrm{rep}}^{\pi}(\mathbf 1)$ and its expected length is $m_{\mathrm{rep}}^{\pi}(\mathbf 1)$, so~\eqref{eq:LTMC pi} follows from the renewal-reward theorem.

\subsection{MTTF-threshold policies}\label{subsec:threshold}

We will focus on MTTF-based repair policies of the following form.
Given a threshold parameter $\gamma$ in $(0,1]$, when we get to a state where the remaining expected working time of the system is less than $100\times \gamma \ \%$ of the MTTF of a brand new system, we repair all failed components; and otherwise we continue operating the system.
That is,
\begin{align}
    \pi^\gamma (\mathbf x) := \begin{cases}
        \actRep & \text{if } m(\mathbf x) / m(\mathbf 1) < \gamma \\
        \actCont & \text{otherwise.}
    \end{cases}
\end{align}

For instance, for the system of Figure~\ref{fig:lemma} with iid exponentially distributed failure times, there are only three distinct policies $\pi^\gamma$, shown in that figure.

We write $LTMC^\gamma := LTMC^{\pi^\gamma}$ for the long-term mean cost~\eqref{eq:LTMC pi} of the threshold policy $\pi^\gamma$.

Solving the linear systems~\eqref{sys:m rep} and~\eqref{sys:c rep} once for each threshold $\gamma$ is expensive: there are as many different policies $\pi^\gamma$ as distinct values of $m(\mathbf x)/m(\mathbf 1)$, which in our experiments range from five to more than a million (Section~\ref{subsec:results thresholds}). We say that the MTTF is \emph{monotone} if $m(\mathbf x') \leq m(\mathbf x)$ for all states $\mathbf x' \leq \mathbf x$, where $m$ is the MTTF~\eqref{eq:mttf} of the system \emph{without} repairs, the quantity that defines the policies $\pi^\gamma$. We now show that, when the MTTF is monotone and known at every state, the long-term mean cost of \emph{all} policies $\pi^\gamma$ is obtained from a single forward pass over the states, without solving any linear system; at the end of this subsection we discuss which failure models have a monotone MTTF.

For a model satisfying Assumptions~(A)--(B), let $f_w$ be the probability that the number of working components, started at $n$, ever takes the value $w$, and let $\varphi_w := \nu_w \sum_{w' < w} (w-w') P_{w,w'}$ be the expected number of component failures per unit time when $w$ components work, i.e., the rate $\nu_w$ of failure epochs times the mean number $\sum_{w' < w} (w-w') P_{w,w'}$ of components that fail at each epoch. For instance, under the L\'evy-frailty Marshall-Olkin distribution of Section~\ref{subsec:models}, failure epochs occur at rate $\nu_w = \Psi(w)$, while $\varphi_w = w \Psi(1)$, since each working component fails at rate $\Psi(1)$, alone or together with others (Remark~\ref{rmk:simultaneity}).
By exchangeability, the expected time that the system without repairs spends in a given state $\mathbf x$ with $w = |\mathbf x|$ working components is
\begin{align}\label{def:tau}
\tau_w := \frac{f_w}{\binom n w \, \nu_w}.
\end{align}

\begin{proposition}\label{prop:all thresholds}
Suppose that Assumptions~(A)--(C) hold and that the MTTF is monotone. For $\gamma$ in $(0,1]$, let $\mathcal C_\gamma := \{\mathbf x : \Phi(\mathbf x) = 1, \ m(\mathbf x) \geq \gamma \, m(\mathbf 1)\}$ be the set of states at which $\pi^\gamma$ continues. Then the expected length $T^\gamma := m^{\pi^\gamma}_{\mathrm{rep}}(\mathbf 1)$ of a repair cycle, the expected number $N^\gamma$ of failed components at the repair, and the probability $P^\gamma$ that the repair is corrective, i.e., that the system has failed, are
\begin{align}
T^\gamma &= \sum_{\mathbf x \in \mathcal C_\gamma} \tau_{|\mathbf x|}, \qquad
N^\gamma = \sum_{\mathbf x \in \mathcal C_\gamma} \tau_{|\mathbf x|} \, \varphi_{|\mathbf x|}, \label{eq:TN gamma} \\
P^\gamma &= \sum_{\mathbf x \in \mathcal C_\gamma} \tau_{|\mathbf x|} \, \nu_{|\mathbf x|} \sum_{w' < |\mathbf x|} P_{|\mathbf x|, w'} \left( 1 - \frac{W^{\mathbf x}_{|\mathbf x| - w'}}{\binom{|\mathbf x|}{|\mathbf x| - w'}} \right), \label{eq:P gamma}
\end{align}
and $c^{\pi^\gamma}_{\mathrm{rep}}(\mathbf 1) = \cset + \ccomp N^\gamma + \csys P^\gamma$, so that
\begin{align}\label{eq:LTMC gamma}
LTMC^\gamma = \frac{\cset + \ccomp \, N^\gamma + \csys \, P^\gamma}{T^\gamma}.
\end{align}
\end{proposition}
The proof, in Appendix~\ref{app:proofs}, uses that, since the MTTF is monotone, the system under $\pi^\gamma$ evolves as the system without repairs until it leaves $\mathcal C_\gamma$.

Proposition~\ref{prop:all thresholds} has two practical consequences. First, both the MTTF~\eqref{eq:E Samaniego sub} and the summands in~\eqref{eq:TN gamma}--\eqref{eq:P gamma} depend on a state $\mathbf x$ only through its number of working components $|\mathbf x|$ and its counting vector $\mathbf W^{\mathbf x}$. Hence we may group the working states into \emph{classes} of states with the same counting vector, and store each class once, with its multiplicity. The number of classes is typically a small fraction of $2^n$: in our experiments, it ranges from $0.005\%$ to $0.4\%$ of the states for the systems with $n \geq 24$ components (Table~\ref{tab:effort}). Second, $\gamma \mapsto LTMC^\gamma$ is piecewise constant, with breakpoints at the distinct values of $m(\mathbf x)/m(\mathbf 1)$, and all of them are evaluated in a single forward pass. Algorithm~\ref{alg}, a backward pass from $\mathbf 0$ to $\mathbf 1$, gives $m(\mathbf x)$ and $\mathbf W^{\mathbf x}$, and hence the summands of~\eqref{eq:TN gamma}--\eqref{eq:P gamma}, at every state; these summands do not depend on $\gamma$. As $\gamma$ decreases from $1$ toward $0$, the set $\mathcal C_\gamma$ grows from the states with $m(\mathbf x) = m(\mathbf 1)$ to all working states, and its states enter in order of decreasing MTTF, which, since the MTTF is monotone, goes from $\mathbf 1$ toward the failed states: no state enters before a larger one. Hence sorting the classes by decreasing MTTF and accumulating their summands gives $T^\gamma$, $N^\gamma$ and $P^\gamma$ at \emph{every} breakpoint as partial sums. These three quantities do not depend on the costs, so that $LTMC^\gamma$ is then available for all thresholds and any costs at the price of a ratio per breakpoint. The curve $\gamma \mapsto LTMC^\gamma$ need not be \emph{unimodal}, i.e., nonincreasing up to some threshold and nondecreasing beyond it, so that once the cost increases it does not decrease again (Section~\ref{subsec:results thresholds}). Hence a search that stops when the cost first increases can miss the best threshold, and evaluating all breakpoints is what makes the best threshold exact.
The weights in~\eqref{eq:E Samaniego sub} and~\eqref{def:tau} are computed with nonnegative terms only: $f_w$ and, for a system of $w$ components, the probability $f^{(w)}_{w-k}$ that its number of working components ever takes the value $w-k$, follow from the forward recursion $f^{(w)}_{v} = \sum_{v' > v} f^{(w)}_{v'} P_{v', v}$ with $f^{(w)}_w = 1$, and then $\mathbb{E} T_{k+1:w} - \mathbb{E} T_{k:w} = f^{(w)}_{w-k} / \nu_{w-k}$.

The monotonicity of the MTTF does not follow from that of $\Phi$ alone, since under Assumptions~(A)--(B) the failure rate of each working component may depend on how many components work. It holds for the L\'evy-frailty Marshall-Olkin distribution~\cite{mai2009levy} of Section~\ref{subsec:models}, which we use in the experiments, and in which each component fails when a common L\'evy subordinator exceeds its own exponentially distributed stress tolerance.

\begin{lemma}\label{lemma:monotone}
If $(T_1, \ldots, T_n)$ follows a L\'evy-frailty Marshall-Olkin distribution, then the MTTF is monotone.
\end{lemma}

The proof, in Appendix~\ref{app:proofs}, uses that, by the memoryless property, the components working at a state behave as a smaller instance of the same model. This fails for sequential order statistics with exponential baseline~\cite{kamps1995concept}, in which components fail one at a time and each of the $w$ working ones fails at rate $\nu_w/w$. For instance, take $n=3$ and $\nu_w = c\,w^2$, with $c = 11/18$ so that each component has unit mean lifetime: if only component~1 is relevant, then $m(1,1,1) = 1 < m(1,0,0) = 1/c \approx 1.64$, since component~1 fails at rate $3c$ while the three components work, but at rate $c$ once it works alone. More generally, the MTTF of this power-trend model $\nu_w = c\,w^q$ is not monotone for any $q > 1$, i.e., whenever the failure rate of each working component decreases as the others fail.

\subsection{Optimal state-based policy}\label{subsec:optimal state}

Under Assumptions~(A) and~(B), and for systems of moderate size, the MTTF-based policies $\pi^\gamma$ can be benchmarked against the optimal policy $\pi^*$ that minimizes the long-term mean cost~\eqref{eq:LTMC pi} over all state-based policies $\pi$ that repair at every failed state ($\Phi(\mathbf x)=0$) and continue at $\mathbf 1$.
With the repair cost $\repcost(\mathbf x)$ of~\eqref{eq:cost setup}, the optimal policy $\pi^*$ and its long-term mean cost $v^\ast_{\mathrm{state}} := LTMC^{\pi^*}$ are obtained with the following iterative algorithm, which is policy iteration for average-cost semi-Markov decision processes~\cite{puterman1994markov,tijms2003first}, or equivalently, Dinkelbach's method for fractional programs~\cite{dinkelbach1967nonlinear}:
\begin{enumerate}
    \item Initialize $k:=0$ and any $g_0 > 0$. Under Assumption~(C), a natural choice is $g_0 := \min_\gamma LTMC^\gamma$.
    \item For all states $\mathbf x$, in increasing order of $|\mathbf x|$, compute:
    \begin{itemize}
      \item if $\Phi(\mathbf x)=0$, fix $h(\mathbf x) := \repcost(\mathbf x)$ and $\pi_{k}(\mathbf x) := \actRep$;
      \item otherwise, if $\mathbf x = \mathbf 1$, fix $\pi_{k}(\mathbf x) := \actCont$ and $h(\mathbf x) := - g_k / \nu_{|\mathbf x|} + \sum_{\mathbf x' \, : \, \mathbf x' < \mathbf x} h(\mathbf x') P_{|\mathbf x|, |\mathbf x'|} / \binom{|\mathbf x|}{|\mathbf x|-|\mathbf x'|}$;
      \item otherwise, fix
      \begin{align*}
        & h(\mathbf x) := \min\Big\{ \repcost(\mathbf x),\; -\frac{g_k}{\nu_{|\mathbf x|}} \\
        & \qquad\qquad\qquad\qquad + \sum_{\mathbf x' \, : \, \mathbf x' \leq \mathbf x}  \frac{h(\mathbf x') P_{|\mathbf x|, |\mathbf x'|}}{\binom{|\mathbf x|}{|\mathbf x|-|\mathbf x'|}} \Big\},
      \end{align*}
      and $\pi_{k}(\mathbf x) := \actRep$ if the minimum is attained by $\repcost(\mathbf x)$, and $\pi_{k}(\mathbf x) := \actCont$ otherwise.
    \end{itemize}
    
    \item For the obtained policy $\pi_k$, solve the linear systems~\eqref{sys:m rep} and~\eqref{sys:c rep} for $m^{\pi_k}_{\mathrm{rep}}$ and $c^{\pi_k}_{\mathrm{rep}}$, respectively, and set $g_{k+1} := c^{\pi_k}_{\mathrm{rep}}(\mathbf 1) / m^{\pi_k}_{\mathrm{rep}}(\mathbf 1)$.
    \item If $h(\mathbf 1) = 0$, stop and return $\pi^* := \pi_k$ and $v^\ast_{\mathrm{state}} := g_{k+1}$; otherwise set $k := k+1$ and go to step 2.
\end{enumerate}
By Assumption~(B), $P_{w,w'}>0$ only for $w'<w$, so the recursion in step~2 is well defined and yields $h(\mathbf 1) = F(g_k)$, where $F(g) := \min_\pi \{ c^\pi_{\mathrm{rep}}(\mathbf 1) - g\, m^\pi_{\mathrm{rep}}(\mathbf 1)\}$, see~\cite{dinkelbach1967nonlinear}.
Moreover, the algorithm terminates after finitely many iterations with $v^\ast_{\mathrm{state}} = \min_\pi c^\pi_{\mathrm{rep}}(\mathbf 1)/m^\pi_{\mathrm{rep}}(\mathbf 1)$.
Indeed, in~\cite{dinkelbach1967nonlinear} it is shown that, for a ratio with positive denominator, the optimal ratio is the unique root of the corresponding parametric function $F$, and the iterates $g_k$ decrease strictly until that root is reached.
Here, $F$ is concave, piecewise linear and strictly decreasing, since it is the minimum of finitely many (one per $\pi$) affine functions $g \mapsto c^\pi_{\mathrm{rep}}(\mathbf 1) - g\, m^\pi_{\mathrm{rep}}(\mathbf 1)$ with slopes $-m^\pi_{\mathrm{rep}}(\mathbf 1) \leq -1/\nu_n < 0$.

The sums over sub-states in step~2 need not be enumerated: as in Lemma~\ref{lemma:Wk}, the \emph{level sums} $H^{\mathbf x}_j := \sum_{\mathbf x' \leq \mathbf x \, : \, |\mathbf x'| = |\mathbf x| - j} h(\mathbf x')$ satisfy $H^{\mathbf x}_0 = h(\mathbf x)$ and $H^{\mathbf x}_j = \frac 1j \sum_{\mathbf x'' \in \mathrm{Sons}(\mathbf x)} H^{\mathbf x''}_{j-1}$ for $j = 1, \ldots, |\mathbf x|$, and the sum in step~2 equals $\sum_{j=1}^{|\mathbf x|} P_{|\mathbf x|, |\mathbf x|-j} H^{\mathbf x}_j / \binom{|\mathbf x|}{j}$. At a failed state all sub-states are failed, so $H^{\mathbf x}_j = \binom{|\mathbf x|}{j} \repcost(\mathbf y)$ for any sub-state $\mathbf y$ with $|\mathbf y| = |\mathbf x|-j$, since $\repcost$ depends on a failed state only through its number of working components. Hence each pass of the algorithm takes $\mathcal O(n^2 2^n)$ operations and keeps only two consecutive levels in memory, as Algorithm~\ref{alg}; the costs and lengths of the current policy needed in step~3 are obtained in the same pass, with the same level sums applied to~\eqref{sys:m rep} and~\eqref{sys:c rep}. Section~\ref{subsec:runtimes} reports the computing times.

\section{Signature dominance and the expressiveness of signature-based policies}\label{sec:dominance}

Algorithm~\ref{alg} computes not only the MTTF at each state, but the whole per-state tail signature $\bar{\mathbf S}^{\mathbf x}$, with $\bar S^{\mathbf x}_k = W^{\mathbf x}_k / \binom{|\mathbf x|}{k}$, from which the MTTF is derived. We now show that the signature is a strictly finer basis for repair decisions. By~\eqref{eq:E Samaniego sub},
\begin{align}\label{eq:m as weighted sum}
    m(\mathbf x) = \sum_{k=0}^{|\mathbf x|-1} \bar S^{\mathbf x}_k \,
    \delta_{k}^{|\mathbf x|},
    \text{ where } \delta_k^{w} := \mathbb{E} T_{k+1:w} - \mathbb{E} T_{k:w}.
\end{align}
The weights $\delta^w_k$ depend on the failure model but not on the system, and they are nonnegative, since the order-statistic means are nondecreasing in their rank; for instance, $\delta^w_k = 1/((w-k)\mu)$ for iid exponential failure times with rate $\mu$, and~\eqref{eq:T LFMO} gives them for the L\'evy-frailty Marshall-Olkin family.

For two states $\mathbf x$ and $\mathbf y$ with $|\mathbf x| = |\mathbf y|$, we say that \emph{$\mathbf x$ signature-dominates $\mathbf y$}, written $\mathbf x \succeq_{\mathrm s} \mathbf y$, if $\bar S^{\mathbf x}_k \geq \bar S^{\mathbf y}_k$ for all $k$; that is, if starting from $\mathbf x$ the system is at least as likely as starting from $\mathbf y$ to survive any given number of further failures.

\begin{theorem}\label{theo:dominance}
Suppose Assumptions~(A), (B) and~(C) hold, and let $\mathbf x, \mathbf y$ be two states with $|\mathbf x| = |\mathbf y| =: w$.
If $\mathbf x \succeq_{\mathrm s} \mathbf y$, then $m(\mathbf x) \geq m(\mathbf y)$. If, moreover, $\bar{\mathbf S}^{\mathbf x} \neq \bar{\mathbf S}^{\mathbf y}$ and $\delta^w_k > 0$ for all $k$, then $m(\mathbf x) > m(\mathbf y)$.
\end{theorem}
The proof is in Appendix~\ref{app:proofs}.

The condition $\delta^w_k > 0$ holds whenever the number of working components, started at $w$, takes every value with positive probability, since $\delta^w_k = f^{(w)}_{w-k} / \nu_{w-k}$ (Section~\ref{subsec:threshold}); this is the case for all the models of the experiments. The converse of Theorem~\ref{theo:dominance} fails, since the tail signatures of two states can cross; comparing states with different numbers of failed components, whose weight vectors have different lengths, is left for future work.

By Theorem~\ref{theo:dominance}, MTTF-threshold policies respect signature dominance: if $\mathbf x \succeq_{\mathrm s} \mathbf y$ and $\pi^\gamma$ repairs at $\mathbf x$, then it also repairs at $\mathbf y$. However, $\pi^\gamma$ acts on a single linear projection of the signature, onto $\boldsymbol\delta^{|\mathbf x|} = (\delta_0^{|\mathbf x|}, \ldots, \delta_{|\mathbf x|-1}^{|\mathbf x|})$, and therefore treats alike all states with the same MTTF, whatever their signatures. Consider instead the \emph{signature-threshold} policies $\pi_{\mathbf u, \gamma}$, which repair at $\mathbf x$ iff $\langle \mathbf u^{|\mathbf x|}, \bar{\mathbf S}^{\mathbf x} \rangle < \gamma \, \langle \mathbf u^{n}, \bar{\mathbf S}^{\mathbf 1} \rangle$, for a weight vector $\mathbf u^v$ for each number $v$ of working components and a threshold $\gamma$. The MTTF-threshold policies are those with $\mathbf u = \boldsymbol \delta$, while other weights can separate states with the same MTTF but different signatures; signature-threshold policies thus strictly refine MTTF-threshold policies.

This is a statement about expressiveness, not about cost. Its value can be bounded: every signature-threshold policy repairs at the failed states, where $\bar{\mathbf S}^{\mathbf x} = \mathbf 0$, and continues at $\mathbf 1$, so it is a state-based policy in the sense of Section~\ref{subsec:optimal state}. Hence the relative gap between the best MTTF-threshold policy and the optimal state-based policy $\pi^\ast$ bounds what \emph{any} signature-threshold policy, whatever its weights, can save over the best MTTF-threshold policy. Section~\ref{sec:experiments} computes this gap.

\section{Computational experiments}\label{sec:experiments}

In this section we first give examples of failure models that satisfy Assumptions~(A)--(C), and describe the failure models, test systems and costs of our experiments. We then report the computational effort of Algorithm~\ref{alg}, the behavior of the MTTF-threshold policies, and their comparison with $r$-out-of-$n$ and optimal state-based policies.

\subsection{Failure models satisfying Assumptions~(A)--(C)}\label{subsec:models}

Under Assumptions~(A)--(B) the survivor count is a pure-death continuous-time Markov chain, so a model is determined by the exit rates $\nu_w$ and the transition kernel $P_{w,w'}$, where $w := |\mathbf x|$ is the number of working components; our algorithm and policies use $\nu_w$, the kernel terms $P_{w, w'} / \binom{w}{w-w'}$ of~\eqref{eq:mttf}, and the increments $\mathbb{E} T_{k+1:w} - \mathbb{E} T_{k:w}$ of~\eqref{eq:E Samaniego sub}. This structure was first identified by Ross, Shahshahani and Weiss~\cite{ross1980number}. Without loss of generality, we calibrate every model so that each component has unit mean lifetime.

\subsubsection{iid exponentially distributed failures}
When the failure times are iid exponentially distributed with rate $\lambda>0$, components fail one at a time, so $P_{w,w'} / \binom{w}{w-w'} = \mathbf{1}_{\{w'=w-1\}}/w$, $\nu_w = w \lambda$ and $\mathbb{E} T_{k+1:w} - \mathbb{E} T_{k:w} = 1 / \nu_{w-k}$. We refer to this model as [Exp].

\subsubsection{L\'evy-frailty Marshall-Olkin (LFMO) model}
In this model of simultaneous failures, the components are subject to a common degradation process, a L\'evy subordinator $\mathbf L$, and each component fails when $\mathbf L$ exceeds its own stress tolerance, which is exponentially distributed with unit rate and independent of everything else; jumps of $\mathbf L$ can make several components fail at once.
The model is characterized by the \emph{Laplace exponent} $\Psi(x) := -\log \mathbb{E} \exp(- x L_1)$ of $\mathbf L$, and it holds that~\cite{lagos2026simple}
\begin{align}\label{eq:nu LFMO}
\nu_w = \Psi(w)
\end{align}
for $0 < w \leq n$ and
\begin{align}\label{eq:P LFMO}
P_{w_1, w_2} \frac{1}{\binom{w_1}{w_1-w_2}} & = \sum_{w=w_2}^{w_1} (-1)^{w-w_2+1} \binom{w_1-w_2}{w_1-w} \frac{\Psi(w)}{\Psi(w_1)}
\end{align}
for $0 \leq w_2 < w_1 \leq n$, and zero otherwise; moreover~\cite{barrera2020limit},
\begin{align}\label{eq:T LFMO}
\mathbb{E} T_{k+1:w} - \mathbb{E} T_{k:w} & = \sum_{i=w-k+1}^w \binom w k \binom k {w-i} (-1)^{i-w+k} \frac 1 {\Psi(i)} .
\end{align}
Hence only $\Psi(1), \ldots, \Psi(n)$ are needed, even for subordinators that cannot be simulated without bias.
The \emph{pure drift} $L_t = \mu t$, with $\Psi(x) = \mu x$, gives the model [Exp] with rate $\mu$.
A compound Poisson process with rate $\lambda$ and iid nonnegative jumps $J_j$ has $\Psi(x) = \lambda (1-\mathbb{E} [e^{-x J_1}])$; we denote by [LFMO-C] the case of exponentially distributed jumps with rate $\theta > 0$, for which $\Psi(x) = \lambda x / (x + \theta)$.
Among the subordinators with infinitely many jumps in finite time, gamma processes, model [LFMO-G], have $\Psi(x) = \beta \log (1+x/\eta)$ for $\beta, \eta > 0$; inverse Gaussian processes, model [LFMO-IG], have $\Psi(x) = \beta (\sqrt{ 2x+ \eta^2}-\eta)$ for $\beta, \eta > 0$; and stable subordinators, model [LFMO-S], have $\Psi(x)=x^\alpha$ for $\alpha$ in $(0,1)$.
Since the sum of independent subordinators is a subordinator whose Laplace exponent is the sum of theirs, a drift can be added to any of them, as a hazard idiosyncratic to each component, while the remaining terms model common shocks.

\begin{remark}[Inverse degree of simultaneity]\label{rmk:simultaneity}
The ratio $\Psi(w)/(w\,\Psi(1))$ admits a direct interpretation as an inverse degree of simultaneity.
Indeed, when there are $w$ working components, the mean number of components that fail at the next failure epoch is $w\,\Psi(1) / \Psi(w)$,
since $\sum_{w_2=0}^{w} (w-w_2) \, P_{w, w_2} = w \, (\Psi(1)-\Psi(0)) / \Psi(w)$ by~\eqref{eq:P LFMO}, with $\Psi(0)=0$.
Moreover, since $\Psi$ is a Bernstein function (see~\cite{schilling2012bernstein}), $x \mapsto \Psi(x)/x$ is non-increasing, so $\Psi(w)/(w\Psi(1))$ lies in $(0,1]$ and equals one exactly when $\Psi$ is linear, i.e., when the subordinator is a pure drift and components fail one at a time.
At the other extreme, when $x \to \infty$, $\Psi(x)/x$ decreases to the drift coefficient $\mu$ of $\mathbf L$, that is, to the share of a component's failure rate that is idiosyncratic rather than shock-driven.
We therefore use $\Psi(n)/(n\Psi(1))$ as a scalar summary of how strongly a given subordinator bundles failures.
Figure~\ref{fig:simultaneity one} plots $w \mapsto \Psi(w)/(w\Psi(1))$ for the failure models of the experiments.
\end{remark}

\subsubsection{Binomial failure rate and lethal-shock models}
In the binomial failure rate model~\cite{vesely1977estimating,atwood1986binomial}, which we denote by [BFR] and which originates in the multivariate exponential distribution of Marshall and Olkin~\cite{marshall1967multivariate}, each component fails on its own at rate $\lambda_{\mathrm{ind}}$, and common-cause shocks arrive at rate $\mu$ and fail each working component independently with probability $p$. In the lethal-shock model~\cite{atwood1986binomial}, which we denote by [LSh], components fail on their own at rate $\lambda_{\mathrm{ind}}$, and lethal shocks arrive at rate $\omega$ and fail all working components at once.
Both are of L\'evy-frailty Marshall-Olkin type, so~\eqref{eq:nu LFMO}--\eqref{eq:T LFMO} apply: [BFR] has $\Psi(x) = \lambda_{\mathrm{ind}}\,x + \mu\big(1-(1-p)^{x}\big)$, a drift plus a compound Poisson subordinator with deterministic jumps of size $\log(1/(1-p))$, where a shock that fails no component leaves the survivor count unchanged~\cite{atwood1986binomial}; and [LSh] has $\Psi(x) = \lambda_{\mathrm{ind}}\,x + \omega$ for $x > 0$, a drift with killing at rate $\omega$. The L\'evy-frailty class thus contains the standard common-cause models of probabilistic safety assessment.

\subsection{Failure models used in the experiments}\label{subsec:models exp}
In the experiments we use the L\'evy-frailty Marshall-Olkin models above, for which the MTTF is monotone (Lemma~\ref{lemma:monotone}): [Exp], and the six models [LFMO-C], [LFMO-G], [LFMO-IG], [LFMO-S], [BFR] and [LSh] with a common calibration.
Models with single failures other than [Exp], such as sequential order statistics~\cite{kamps1995concept} and load-sharing models~\cite{ross1984model}, are not used, since they do not admit the calibration below and the MTTF of some of them is not monotone (Section~\ref{subsec:threshold}).
Every model is normalized so that $\Psi(1)=1$; since each component's lifetime is exponentially distributed with rate $\Psi(1)$, one time unit is then the mean lifetime of a component, and all models differ only in how failures are grouped.
For [LFMO-C] and [BFR] we take driftless compound Poisson subordinators, i.e., $\mu = 0$ and $\lambda_{\mathrm{ind}} = 0$, respectively, so that, after the normalization, every model other than [Exp] has a single free parameter, which we fix through the pairwise beta factor.

\paragraph{Calibration by the pairwise beta factor}
In probabilistic safety assessment, the dependence between redundant components is most often quantified by the \emph{beta factor} of Fleming~\cite{fleming1975reliability}, still the default common-cause failure model in standards such as IEC~61508~\cite{iec61508part6}; see~\cite{mosleh1998guidelines}. The failure rate $\lambda$ of each component of a redundant group is split as $\lambda = \lambda_{\mathrm{ind}} + \lambda_{\mathrm{ccf}}$, into independent failures and common-cause events that fail the whole group, and the beta factor is the fraction $\lambda_{\mathrm{ccf}}/\lambda$.
For two components $i \neq j$ of a L\'evy-frailty Marshall-Olkin system, $\Pr(T_i > t, T_j > t) = e^{-t \Psi(2)}$ and $\Pr(T_i > t) = e^{-t \Psi(1)}$, so the pair follows the bivariate Marshall-Olkin distribution~\cite{marshall1967multivariate} in which each component fails alone at rate $\Psi(2) - \Psi(1)$ and both fail together at rate $2\Psi(1) - \Psi(2)$. We therefore define the \emph{pairwise beta factor}
\begin{align}\label{def:beta}
\beta := \frac{2 \Psi(1) - \Psi(2)}{\Psi(1)} = 2 - \Psi(2) \qquad \text{when } \Psi(1) = 1,
\end{align}
the fraction of the failure rate of a component due to events that also fail a given other component. It lies in $[0,1]$, since $\Psi$ is concave and non-decreasing with $\Psi(0)=0$, and $\beta = 0$ exactly for [Exp].
Calibrating every model to the same $\beta$ makes all of them agree on every \emph{pair} of components, so that they differ only in how simultaneity grows with the number of components involved. In fact, [LSh] is the beta-factor model applied to the whole system, with $\beta = \omega$; [BFR] replaces its lethal shock by a binomial one~\cite{atwood1986binomial}; and the four [LFMO] models extend this to random shock sizes and to infinitely many small shocks.
Table~\ref{tab:models} gives the resulting Laplace exponents and parameters; each model has a unique parameter for each $\beta$ in $(0,1)$, except [LFMO-IG], whose pairwise beta factor cannot exceed $2 - \sqrt 2 \approx 0.586$.
We use $\beta \in \{1\%, 5\%, 10\%, 20\%, 50\%\}$, and we call $\beta = 5\%$, $20\%$ and $50\%$ the \emph{realistic}, \emph{poorly separated} and \emph{stressed} scenarios: generic beta factors are of the order of $1\%$ to $10\%$ for well-designed redundant equipment~\cite{iec61508part6}, and larger values arise when the components share their environment, design or maintenance.
Together with [Exp], this gives $1 + 6 \times 5 = 31$ failure models.

By Remark~\ref{rmk:simultaneity}, the mean number $w/\Psi(w)$ of components that fail together when $w$ components work grows without bound for the driftless models: linearly in $w$ for [LFMO-C] and [BFR], as $w/\log w$ for [LFMO-G], as $\sqrt w$ for [LFMO-IG] and as $w^{1-\alpha}$ for [LFMO-S]; for [LSh], instead, it is bounded by $1/(1-\beta)$.
Figure~\ref{fig:simultaneity one} shows $\Psi(w)/w$ for all models. For the systems with $n = 30$ components, the mean number of components failing at the first failure epoch, $n/\Psi(n)$, ranges from $1.05$ ([LSh]) to $1.91$ ([BFR]) in the realistic scenario, from $1.24$ to $6.0$ in the poorly separated one, and from $1.94$ to $15.0$ in the stressed one: all models have the same pairwise dependence, but they produce very different batches of simultaneous failures in large systems.

\definecolor{cExp}    {HTML}{1F77B4}
\definecolor{cLFMOC}  {HTML}{9467BD}
\definecolor{cLFMOG}  {HTML}{2CA02C}
\definecolor{cLFMOIG} {HTML}{8C564B}
\definecolor{cLFMOS}  {HTML}{FF7F0E}
\definecolor{cBFR}    {HTML}{E377C2}
\definecolor{cLSh}    {HTML}{17BECF}

\begin{table*}[!t]

\centering
\caption{
Failure models used in the experiments, all of L\'evy-frailty Marshall-Olkin type and calibrated so that $\Psi(1)=1$, i.e., each component has unit mean lifetime.
For each model we give its Laplace exponent $\Psi(x)$ for $x>0$ (with $\Psi(0):=0$), its free parameter as a function of the pairwise beta factor $\beta = 2-\Psi(2)$ of~\eqref{def:beta}, its value in the realistic ($\beta=5\%$), poorly separated ($\beta=20\%$) and stressed ($\beta=50\%$) scenarios, and the inverse degree of simultaneity $\Psi(w)/w$ of Remark~\ref{rmk:simultaneity}.
All models except [Exp] and [LSh] are driftless.
}
\label{tab:models}
\small
\renewcommand{\arraystretch}{1.4}
\setlength{\tabcolsep}{4pt}
\begin{tabular}{@{}lllcccl@{}}
\toprule
& & & \multicolumn{3}{c}{Parameter} & \\
\cmidrule(lr){4-6}
Label & $\Psi(x)$ & Parameter & $\beta=5\%$ & $\beta=20\%$ & $\beta=50\%$ & $\Psi(w)/w$ \\
\midrule
$[$Exp$]$     & $x$ & --- & --- & --- & --- & $1$ \\
$[$LFMO-C$]$  & $\dfrac{x}{1+p(x-1)}$ & $p = \dfrac{\beta}{2-\beta}$ & $0.0256$ & $0.111$ & $0.333$ & $\dfrac{1}{1+p(w-1)}$ \\
$[$LFMO-G$]$  & $\dfrac{\log(1+x/\eta)}{\log(1+1/\eta)}$ & $\eta$: $\Psi(2) = 2-\beta$ & $18.51$ & $3.537$ & $0.618$ & $\dfrac{\log(1+w/\eta)}{w\log(1+1/\eta)}$ \\
$[$LFMO-IG$]$ & $\dfrac{\sqrt{2x+\eta^2}-\eta}{\sqrt{2+\eta^2}-\eta}$ & $\eta$: $\Psi(2) = 2-\beta$ & $4.188$ & $1.634$ & $0.289$ & $\dfrac{\sqrt{2w+\eta^2}-\eta}{w(\sqrt{2+\eta^2}-\eta)}$ \\
$[$LFMO-S$]$  & $x^\alpha$ & $\alpha = \log_2(2-\beta)$ & $0.963$ & $0.848$ & $0.585$ & $w^{\alpha-1}$ \\
$[$BFR$]$     & $\dfrac{1-(1-p)^x}{p}$ & $p = \beta$ & $0.05$ & $0.2$ & $0.5$ & $\dfrac{1-(1-p)^w}{p\,w}$ \\
$[$LSh$]$     & $(1-\theta)x + \theta$ & $\theta = \beta$ & $0.05$ & $0.2$ & $0.5$ & $1-\theta+\dfrac{\theta}{w}$ \\
\bottomrule
\end{tabular}
\end{table*}

\pgfmathdeclarefunction{kC}{2}{\pgfmathparse{1/(1+#2*(#1-1))}}
\pgfmathdeclarefunction{kG}{2}{\pgfmathparse{ln(1+#1/#2)/(#1*ln(1+1/#2))}}
\pgfmathdeclarefunction{kIG}{2}{\pgfmathparse{(sqrt(2*#1+#2*#2)-#2)/(#1*(sqrt(2+#2*#2)-#2))}}
\pgfmathdeclarefunction{kS}{2}{\pgfmathparse{#1^(#2-1)}}
\pgfmathdeclarefunction{kBFR}{2}{\pgfmathparse{(1-(1-#2)^#1)/(#2*#1)}}
\pgfmathdeclarefunction{kLSh}{2}{\pgfmathparse{1-#2+#2/#1}}
\def\pC{{0.025641,0.111111,0.333333}}
\def\pG{{18.5086,3.53728,0.618034}}
\def\pIG{{4.1879,1.63384,0.288675}}
\def\pS{{0.963474,0.847997,0.584963}}
\def\pBFR{{0.05,0.2,0.5}}
\def\pLSh{{0.05,0.2,0.5}}
\newcommand{\yehticks}{\addplot[only marks, mark=|, mark size=2pt, black, forget plot]
  coordinates {(5,0) (8,0) (9,0) (12,0) (13,0) (14,0) (21,0) (22,0) (23,0) (24,0) (26,0) (27,0) (30,0)};}
\pgfplotsset{kappaaxis/.style={
  xmin=1, xmax=30, ymin=0, ymax=1.05, domain=1:30, samples=59,
  xtick={1,5,10,15,20,25,30}, ytick={0,0.25,0.5,0.75,1},
  xlabel={$w$}, ylabel={$\Psi(w)/w$},
  grid=both, grid style={line width=0.3pt, draw=gray!25},
  tick label style={font=\footnotesize}, label style={font=\footnotesize}}}
\newcommand{\kappalines}[3]{%
  \pgfmathsetmacro{\kLo}{#3[0]}\pgfmathsetmacro{\kMid}{#3[1]}\pgfmathsetmacro{\kHi}{#3[2]}%
  \addplot[#1, line width=1pt, dashed, forget plot] {#2(x,\kLo)};
  \addplot[#1, line width=1pt, forget plot] {#2(x,\kMid)};
  \addplot[#1, line width=1.2pt, densely dotted, forget plot] {#2(x,\kHi)};
}
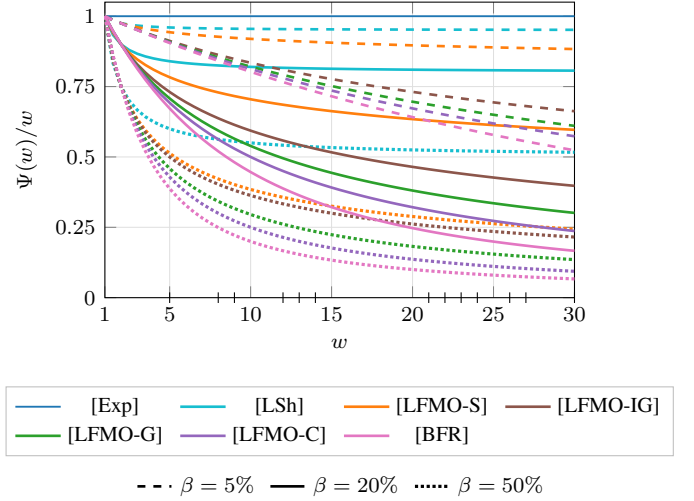
\begin{figure}[!t]

\centering
\begin{tikzpicture}
\begin{axis}[
  kappaaxis, width=0.88\linewidth, height=0.62\linewidth,
  legend style={font=\footnotesize, at={(0.5,-0.3)}, anchor=north, legend columns=4,
                draw=gray!50, /tikz/every even column/.append style={column sep=4pt}},
]
\addplot[cExp, line width=0.7pt] {1}; \addlegendentry{[Exp]}
\kappalines{cLSh}{kLSh}{\pLSh}      \addlegendimage{cLSh, line width=1pt}    \addlegendentry{[LSh]}
\kappalines{cLFMOS}{kS}{\pS}        \addlegendimage{cLFMOS, line width=1pt}  \addlegendentry{[LFMO-S]}
\kappalines{cLFMOIG}{kIG}{\pIG}     \addlegendimage{cLFMOIG, line width=1pt} \addlegendentry{[LFMO-IG]}
\kappalines{cLFMOG}{kG}{\pG}        \addlegendimage{cLFMOG, line width=1pt}  \addlegendentry{[LFMO-G]}
\kappalines{cLFMOC}{kC}{\pC}        \addlegendimage{cLFMOC, line width=1pt}  \addlegendentry{[LFMO-C]}
\kappalines{cBFR}{kBFR}{\pBFR}      \addlegendimage{cBFR, line width=1pt}    \addlegendentry{[BFR]}
\yehticks
\end{axis}
\end{tikzpicture}

\vspace{1ex}
{\footnotesize
\tikz[baseline=-0.6ex]\draw[black, line width=1pt, dashed] (0,0) -- (0.45,0); $\beta=5\%$\ \
\tikz[baseline=-0.6ex]\draw[black, line width=1pt] (0,0) -- (0.45,0); $\beta=20\%$\ \
\tikz[baseline=-0.6ex]\draw[black, line width=1.2pt, densely dotted] (0,0) -- (0.45,0); $\beta=50\%$\par}
\caption{
Inverse degree of simultaneity $\Psi(w)/w$ (Remark~\ref{rmk:simultaneity}) of the failure models of Table~\ref{tab:models}, as a function of the number $w$ of working components; its reciprocal $w/\Psi(w)$ is the mean number of components that fail together at the next failure epoch. Each model has its own colour, and the dashed, solid and densely dotted lines are the realistic ($\beta=5\%$), poorly separated ($\beta=20\%$) and stressed ($\beta=50\%$) scenarios, respectively.
The thin line $\Psi(w)/w \equiv 1$ is [Exp], where failures are never simultaneous.
All models coincide at $w=1$ and, at each level of $\beta$, at $w=2$, since $\Psi(2)/2 = 1-\beta/2$.
They differ in how fast simultaneity grows with the number of working components.
The ticks on the horizontal axis mark the sizes $n$ of the systems of Figure~\ref{fig:yeh}.
}
\label{fig:simultaneity one}
\end{figure}

\subsection{Systems used in the experiments}\label{subsec:systems}

We test our results on the 22 benchmark networks that Yeh et al.~\cite{yeh2024new} use to compare algorithms for finding all minimal cuts, shown in Figure~\ref{fig:yeh}.
Each of them is a \emph{two-terminal network}: a graph with a source $s$ and a terminal $t$, whose edges are the components of the system and whose nodes do not fail. The system works iff $s$ and $t$ are connected by working edges, i.e., $\Phi(\mathbf x) = 1$ iff the edges $i$ with $x_i = 1$ contain a path from $s$ to $t$. This structure function is monotone, with $\Phi(\mathbf 0)=0$ and $\Phi(\mathbf 1)=1$, and the minimal path sets are the edge sets of the simple paths from $s$ to $t$.
Hence the number $n$ of components of a system is its number of \emph{edges}, not of nodes; it ranges from $5$ to $30$, and $2^n$ from $32$ to about $1.07 \times 10^9$.

We use the edge lists of all 22 networks as they appear in Figure~7 of~\cite{yeh2024new}, which leads to three features worth noting.
First, networks 6 and 7 are drawn with arrows, and we read every edge of these two networks as an arc that can only be traversed in its direction. Network 6 then has four irrelevant components, the arcs leaving node~5, which cannot be reached from $s$, so it is semi-coherent but not coherent; it has $13$ minimal cuts, and network 7 has $110$, whereas, read as undirected graphs, they would have $26$ and $133$, the value reported for network 7 in Table~9 of~\cite{yeh2024new}.
All other networks are undirected, and all their components are relevant.
Second, network 21 is network 13 redrawn with the nodes renumbered, as noted in Section~6.3 of~\cite{yeh2024new}: the two are isomorphic, with the terminals fixed, so they have the same signatures, MTTFs and policies, and we use network 21 only as a check of our computations (all its results coincide with those of network 13).
Third, network 22 is network 18 plus one edge, so it has $n = 23$ components; its drawing in Figure~7 of~\cite{yeh2024new} indeed has 23 edges, although their Table~9 reports 22.
For the other networks, the numbers of nodes, edges and minimal cuts we obtain coincide with Table~9 of~\cite{yeh2024new}, except for network 10, whose row there repeats that of network 15, and network 18, for which we find 219 minimal cuts instead of 214, even though our edge list coincides with their figure.

\newcommand{\netsc}{0.48}                     %
\newlength{\netnd}\setlength{\netnd}{2.1mm}   %
\tikzset{
  net/.style={
    nd/.style={circle, draw=black!75, line width=0.25pt, fill=blue!7,
               inner sep=0pt, minimum size=\netnd,
               font=\fontsize{3.6}{4}\selectfont},
    ed/.style={draw=black!75, line width=0.25pt},
    st/.style={font=\fontsize{5}{5.5}\selectfont, inner sep=0pt}
  },
  netb/.style={nd/.append style={minimum size=1.25\netnd,
               font=\fontsize{4.3}{4.8}\selectfont}}
}

\captionsetup[subfigure]{labelformat=simple}
\renewcommand{\thesubfigure}{(\arabic{subfigure})}
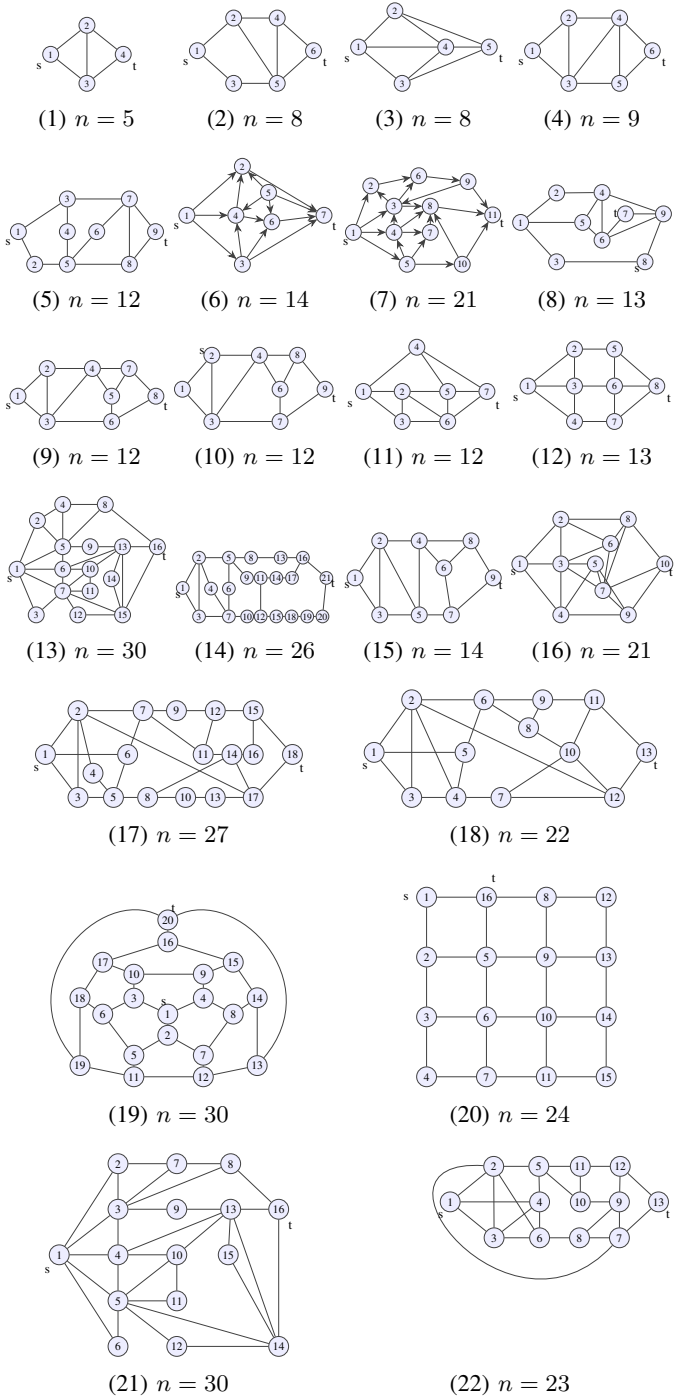
\begin{figure}[!t]
\centering

\begin{subfigure}[b]{0.245\linewidth}\centering
\begin{tikzpicture}[net,scale=\netsc*1.0]
  \foreach \n/\x/\y in {1/0/0, 2/1/0.8, 3/1/-0.8, 4/2/0}
     \node[nd] (n\n) at (\x,\y) {\n};
  \foreach \a/\b in {1/2,1/3,2/3,2/4,3/4} \draw[ed] (n\a)--(n\b);
  \node[st] at ([xshift=-3.2mm,yshift=-3.2mm]n1) {s};
  \node[st] at ([xshift=3.2mm,yshift=-3.2mm]n4) {t};
\end{tikzpicture}
\caption{$n=5$}\label{fig:tn1}%
\end{subfigure}\hfill
\begin{subfigure}[b]{0.245\linewidth}\centering
\begin{tikzpicture}[net,scale=\netsc*1.0]
  \foreach \n/\x/\y in {1/0/0, 2/1/0.9, 3/1/-0.9, 4/2.2/0.9, 5/2.2/-0.9, 6/3.2/0}
     \node[nd] (n\n) at (\x,\y) {\n};
  \foreach \a/\b in {1/2,1/3,2/4,2/5,3/5,4/5,4/6,5/6} \draw[ed] (n\a)--(n\b);
  \node[st] at ([xshift=-3.2mm,yshift=-3.2mm]n1) {s};
  \node[st] at ([xshift=3.2mm,yshift=-3.2mm]n6) {t};
\end{tikzpicture}
\caption{$n=8$}\label{fig:tn2}%
\end{subfigure}\hfill
\begin{subfigure}[b]{0.245\linewidth}\centering
\begin{tikzpicture}[net,scale=\netsc*1.0]
  \foreach \n/\x/\y in {1/0/0, 2/1/1, 3/1.2/-1, 4/2.4/0, 5/3.6/0}
     \node[nd] (n\n) at (\x,\y) {\n};
  \foreach \a/\b in {1/2,1/3,1/4,2/4,2/5,3/4,3/5,4/5} \draw[ed] (n\a)--(n\b);
  \node[st] at ([xshift=-3.2mm,yshift=-3.2mm]n1) {s};
  \node[st] at ([xshift=3.2mm,yshift=-3.2mm]n5) {t};
\end{tikzpicture}
\caption{$n=8$}\label{fig:tn3}%
\end{subfigure}\hfill
\begin{subfigure}[b]{0.245\linewidth}\centering
\begin{tikzpicture}[net,scale=\netsc*1.0]
  \foreach \n/\x/\y in {1/0/0, 2/1/0.9, 3/1/-0.9, 4/2.4/0.9, 5/2.4/-0.9, 6/3.3/0}
     \node[nd] (n\n) at (\x,\y) {\n};
  \foreach \a/\b in {1/2,1/3,2/3,2/4,3/4,3/5,4/5,4/6,5/6} \draw[ed] (n\a)--(n\b);
  \node[st] at ([xshift=-3.2mm,yshift=-3.2mm]n1) {s};
  \node[st] at ([xshift=3.2mm,yshift=-3.2mm]n6) {t};
\end{tikzpicture}
\caption{$n=9$}\label{fig:tn4}%
\end{subfigure}

\vspace{2ex}

\begin{subfigure}[b]{0.245\linewidth}\centering
\begin{tikzpicture}[net,scale=\netsc*0.90]
  \foreach \n/\x/\y in {1/0/0, 2/0.5/-1, 3/1.5/1, 4/1.5/0, 5/1.5/-1,
                        6/2.4/0, 7/3.4/1, 8/3.4/-1, 9/4.2/0}
     \node[nd] (n\n) at (\x,\y) {\n};
  \foreach \a/\b in {1/2,1/3,2/5,3/4,3/7,4/5,5/6,5/8,6/7,7/8,7/9,8/9}
     \draw[ed] (n\a)--(n\b);
  \node[st] at ([xshift=-3.2mm,yshift=-3.2mm]n1) {s};
  \node[st] at ([xshift=3.2mm,yshift=-3.2mm]n9) {t};
\end{tikzpicture}
\caption{$n=12$}\label{fig:tn5}%
\end{subfigure}\hfill
\begin{subfigure}[b]{0.245\linewidth}\centering
\begin{tikzpicture}[net,scale=\netsc*0.90]
  \foreach \n/\x/\y in {1/0/0, 2/1.7/1.5, 3/1.7/-1.5, 4/1.5/0,
                        5/2.5/0.7, 6/2.6/-0.2, 7/4.2/0}
     \node[nd] (n\n) at (\x,\y) {\n};
  \foreach \a/\b in {1/2,1/3,1/4,2/7,3/4,3/6,3/7,4/2,4/6,5/2,5/4,5/6,5/7,6/7}
     \draw[ed,-{Stealth[length=1.1mm,width=0.9mm]}] (n\a)--(n\b);%
  \node[st] at ([xshift=-3.2mm,yshift=-3.2mm]n1) {s};
  \node[st] at ([xshift=3.2mm,yshift=-3.2mm]n7) {t};
\end{tikzpicture}
\caption{$n=14$}\label{fig:tn6}%
\end{subfigure}\hfill
\begin{subfigure}[b]{0.245\linewidth}\centering
\begin{tikzpicture}[net,scale=\netsc*0.80]
  \foreach \n/\x/\y in {1/0/1.5, 2/0.6/3.1, 3/1.4/2.4, 4/1.4/1.5, 5/1.85/0.4,
                        6/2.25/3.45, 7/2.65/1.5, 8/2.65/2.4, 9/3.95/3.25,
                        10/3.75/0.4, 11/4.8/2.1}
     \node[nd] (n\n) at (\x,\y) {\n};
  \foreach \a/\b in {1/2,1/3,1/4,1/5,2/6,3/2,3/6,3/8,4/3,4/7,4/8,5/4,5/7,5/10,
                     6/9,7/8,9/3,9/11,8/11,10/8,10/11}
     \draw[ed,-{Stealth[length=1.1mm,width=0.9mm]}] (n\a)--(n\b);%
  \node[st] at ([xshift=-2.6mm,yshift=-3.2mm]n1) {s};
  \node[st] at ([xshift=3.2mm,yshift=-3.0mm]n11) {t};
\end{tikzpicture}
\caption{$n=21$}\label{fig:tn7}%
\end{subfigure}\hfill
\begin{subfigure}[b]{0.245\linewidth}\centering
\begin{tikzpicture}[net,scale=\netsc*0.62]
  \foreach \n/\x/\y in {1/0/2.27, 2/1.6/3.58, 3/1.6/0.52, 4/3.62/3.6, 5/2.76/2.27,
                        6/3.62/1.48, 7/4.66/2.65, 8/5.53/0.52, 9/6.35/2.65}
     \node[nd] (n\n) at (\x,\y) {\n};
  \foreach \a/\b in {1/2,1/5,1/3,2/4,4/5,4/6,4/9,5/6,6/7,6/9,7/9,3/8,8/9}
     \draw[ed] (n\a)--(n\b);
  \node[st] at ([xshift=-3.6mm,yshift=-3.4mm]n8) {s};
  \node[st] at ([xshift=-4.3mm,yshift=0.2mm]n7) {t};
\end{tikzpicture}
\caption{$n=13$}\label{fig:tn8}%
\end{subfigure}

\vspace{2ex}

\begin{subfigure}[b]{0.245\linewidth}\centering
\begin{tikzpicture}[net,scale=\netsc*0.83]
  \foreach \n/\x/\y in {1/0/0.92, 2/0.97/1.84, 3/0.97/0.08, 4/2.47/1.84,
                        5/3.11/0.92, 6/3.11/0.08, 7/3.69/1.84, 8/4.58/0.92}
     \node[nd] (n\n) at (\x,\y) {\n};
  \foreach \a/\b in {1/2,1/3,2/3,2/4,3/4,3/6,4/5,4/7,5/6,5/7,6/8,7/8}
     \draw[ed] (n\a)--(n\b);
  \node[st] at ([xshift=-3.2mm,yshift=-3.2mm]n1) {s};
  \node[st] at ([xshift=3.2mm,yshift=-3.2mm]n8) {t};
\end{tikzpicture}
\caption{$n=12$}\label{fig:tn9}%
\end{subfigure}\hfill
\begin{subfigure}[b]{0.245\linewidth}\centering
\begin{tikzpicture}[net,scale=\netsc*0.78]
  \foreach \n/\x/\y in {1/0/1.48, 2/1.04/2.68, 3/1.04/0.34, 4/2.72/2.68,
                        6/3.45/1.48, 7/3.45/0.34, 8/4.07/2.68, 9/5.04/1.48}
     \node[nd] (n\n) at (\x,\y) {\n};
  \foreach \a/\b in {1/2,1/3,2/3,2/4,3/4,3/7,4/6,4/8,6/7,6/8,7/9,8/9}
     \draw[ed] (n\a)--(n\b);
  \node[st] at ([xshift=-3.4mm,yshift=1.6mm]n2) {s};
  \node[st] at ([xshift=3.2mm,yshift=-3.2mm]n9) {t};
\end{tikzpicture}
\caption{$n=12$}\label{fig:tn10}%
\end{subfigure}\hfill
\begin{subfigure}[b]{0.245\linewidth}\centering
\begin{tikzpicture}[net,scale=\netsc*1.15]
  \foreach \n/\x/\y in {1/0/1.05, 2/0.93/1.05, 3/0.93/0.33, 4/1.28/2.11,
                        5/2.02/1.05, 6/2.02/0.33, 7/2.96/1.05}
     \node[nd] (n\n) at (\x,\y) {\n};
  \foreach \a/\b in {1/4,1/2,1/3,2/3,2/5,2/6,3/6,4/5,4/7,5/6,5/7,6/7}
     \draw[ed] (n\a)--(n\b);
  \node[st] at ([xshift=-3.2mm,yshift=-3.0mm]n1) {s};
  \node[st] at ([xshift=2.6mm,yshift=-3.2mm]n7) {t};
\end{tikzpicture}
\caption{$n=12$}\label{fig:tn11}%
\end{subfigure}\hfill
\begin{subfigure}[b]{0.245\linewidth}\centering
\begin{tikzpicture}[net,scale=\netsc*1.05]
  \foreach \n/\x/\y in {1/0/1.31, 2/1.22/2.28, 3/1.22/1.31, 4/1.22/0.38,
                        5/2.27/2.28, 6/2.27/1.31, 7/2.27/0.38, 8/3.38/1.31}
     \node[nd] (n\n) at (\x,\y) {\n};
  \foreach \a/\b in {1/2,1/3,1/4,2/3,2/5,3/4,3/6,4/7,5/6,5/8,6/7,6/8,7/8}
     \draw[ed] (n\a)--(n\b);
  \node[st] at ([xshift=-3.2mm,yshift=-3.2mm]n1) {s};
  \node[st] at ([xshift=3.2mm,yshift=-3.2mm]n8) {t};
\end{tikzpicture}
\caption{$n=13$}\label{fig:tn12}%
\end{subfigure}

\vspace{2ex}

\begin{subfigure}[b]{0.245\linewidth}\centering
\begin{tikzpicture}[net,scale=\netsc*0.85]
  \foreach \n/\x/\y in {1/0/1.81, 2/0.66/3.39, 3/0.62/0.33, 4/1.48/3.91,
                        5/1.48/2.54, 6/1.48/1.81, 7/1.48/1.10, 8/2.86/3.91,
                        9/2.37/2.54, 10/2.37/1.81, 11/2.37/1.10, 12/1.97/0.33,
                        13/3.42/2.54, 14/3.06/1.50, 15/3.42/0.33, 16/4.55/2.54}
     \node[nd] (n\n) at (\x,\y) {\n};
  \foreach \a/\b in {1/2,1/3,1/5,1/6,1/7,2/4,2/5,4/5,4/8,5/6,5/8,5/9,
                     6/7,6/10,6/13,7/3,7/10,7/11,7/12,7/15,8/16,9/13,
                     10/11,10/13,12/15,13/14,13/15,13/16,14/15,15/16}
     \draw[ed] (n\a)--(n\b);
  \node[st] at ([xshift=-2.4mm,yshift=-3.2mm]n1) {s};
  \node[st] at ([xshift=2.2mm,yshift=-3.2mm]n16) {t};
\end{tikzpicture}
\caption{$n=30$}\label{fig:tn13}%
\end{subfigure}\hfill
\begin{subfigure}[b]{0.245\linewidth}\centering
\begin{tikzpicture}[net,nd/.append style={minimum size=0.80\netnd},scale=\netsc*0.60]
  \foreach \n/\x/\y in {1/0/1.73, 2/0.75/3.12, 3/0.75/0.42, 4/1.30/1.70,
                        5/2.12/3.12, 6/2.12/1.70, 7/2.12/0.42, 8/3.14/3.12,
                        9/2.95/2.19, 10/2.96/0.42, 11/3.58/2.19, 12/3.58/0.42,
                        13/4.50/3.12, 14/4.30/2.19, 15/4.30/0.42, 16/5.52/3.12,
                        17/5.02/2.19, 18/5.00/0.42, 19/5.72/0.42, 20/6.42/0.42,
                        21/6.60/2.19}
     \node[nd] (n\n) at (\x,\y) {\n};
  \foreach \a/\b in {1/2,1/3,2/3,2/4,2/5,3/7,4/7,5/6,5/8,5/9,6/7,7/10,
                     8/13,9/11,10/12,11/12,11/14,12/15,13/16,14/17,15/18,
                     16/17,16/21,18/19,19/20,20/21}
     \draw[ed] (n\a)--(n\b);
  \node[st] at ([xshift=-2.2mm,yshift=-3.2mm]n1) {s};
  \node[st] at ([xshift=3.0mm,yshift=-2.6mm]n21) {t};
\end{tikzpicture}
\caption{$n=26$}\label{fig:tn14}%
\end{subfigure}\hfill
\begin{subfigure}[b]{0.245\linewidth}\centering
\begin{tikzpicture}[net,scale=\netsc*0.68]
  \foreach \n/\x/\y in {1/0/1.5, 2/1/3, 3/1/0, 4/2.6/3, 5/2.6/0,
                        6/3.6/1.9, 7/3.9/0, 8/4.7/3, 9/5.6/1.5}
     \node[nd] (n\n) at (\x,\y) {\n};
  \foreach \a/\b in {1/2,1/3,2/3,2/4,2/5,3/5,4/5,4/6,4/8,5/7,6/7,6/8,7/9,8/9}
     \draw[ed] (n\a)--(n\b);
  \node[st] at ([xshift=-3.2mm,yshift=-3.0mm]n1) {s};
  \node[st] at ([xshift=3.0mm,yshift=-3.0mm]n9) {t};
\end{tikzpicture}
\caption{$n=14$}\label{fig:tn15}%
\end{subfigure}\hfill
\begin{subfigure}[b]{0.245\linewidth}\centering
\begin{tikzpicture}[net,scale=\netsc*0.72]
  \foreach \n/\x/\y in {1/0/2.20, 2/1.57/3.91, 3/1.57/2.20, 4/1.57/0.25,
                        5/2.89/2.20, 6/3.47/2.98, 7/3.18/1.18, 8/4.15/3.91,
                        9/4.15/0.25, 10/5.55/2.20}
     \node[nd] (n\n) at (\x,\y) {\n};
  \foreach \a/\b in {1/2,1/3,1/4,2/3,2/6,2/8,3/4,3/5,3/6,3/7,4/5,4/9,
                     5/7,5/9,6/7,6/8,7/8,7/9,7/10,8/10,9/10}
     \draw[ed] (n\a)--(n\b);
  \node[st] at ([xshift=-2.4mm,yshift=-3.2mm]n1) {s};
  \node[st] at ([xshift=2.4mm,yshift=-3.2mm]n10) {t};
\end{tikzpicture}
\caption{$n=21$}\label{fig:tn16}%
\end{subfigure}

\vspace{2ex}

\begin{subfigure}[b]{0.49\linewidth}\centering
\begin{tikzpicture}[net,netb,scale=\netsc*1.15]
  \foreach \n/\x/\y in {1/0/1.34, 2/0.78/2.39, 3/0.78/0.31, 4/1.14/0.89,
                        5/1.63/0.31, 6/1.97/1.34, 7/2.34/2.39, 8/2.45/0.31,
                        9/3.14/2.39, 10/3.36/0.31, 11/3.78/1.34, 12/4.07/2.39,
                        13/4.07/0.31, 14/4.47/1.34, 15/4.98/2.39, 16/4.98/1.34,
                        17/4.98/0.31, 18/5.92/1.34}
     \node[nd] (n\n) at (\x,\y) {\n};
  \foreach \a/\b in {1/2,1/3,1/6,2/3,2/4,2/7,2/17,3/5,4/5,5/6,5/8,6/7,
                     7/9,7/11,8/10,8/14,9/12,10/13,11/12,11/14,12/15,
                     13/17,14/16,14/17,15/16,15/18,17/18}
     \draw[ed] (n\a)--(n\b);
  \node[st] at ([xshift=-2.2mm,yshift=-3.2mm]n1) {s};
  \node[st] at ([xshift=2.2mm,yshift=-3.2mm]n18) {t};
\end{tikzpicture}
\caption{$n=27$}\label{fig:tn17}%
\end{subfigure}\hfill
\begin{subfigure}[b]{0.49\linewidth}\centering
\begin{tikzpicture}[net,netb,scale=\netsc*1.15]
  \foreach \n/\x/\y in {1/0/1.40, 2/0.90/2.66, 3/0.90/0.32, 4/1.96/0.32,
                        5/2.18/1.40, 6/2.63/2.66, 7/3.04/0.32, 8/3.70/1.99,
                        9/4.04/2.66, 10/4.69/1.40, 11/5.28/2.66, 12/5.76/0.32,
                        13/6.53/1.40}
     \node[nd] (n\n) at (\x,\y) {\n};
  \foreach \a/\b in {1/2,1/3,1/5,2/3,2/4,2/6,2/12,3/4,4/5,4/7,5/6,6/8,6/9,
                     7/10,7/12,8/9,8/10,9/11,10/11,10/12,11/13,12/13}
     \draw[ed] (n\a)--(n\b);
  \node[st] at ([xshift=-2.2mm,yshift=-3.2mm]n1) {s};
  \node[st] at ([xshift=2.2mm,yshift=-3.2mm]n13) {t};
\end{tikzpicture}
\caption{$n=22$}\label{fig:tn18}%
\end{subfigure}

\vspace{2ex}

\begin{subfigure}[b]{0.49\linewidth}\centering
\begin{tikzpicture}[net,netb,scale=\netsc*1.10]
  \foreach \n/\x/\y in {1/2.20/1.78, 2/2.20/1.25, 3/1.35/2.18, 4/3.09/2.18,
                        5/1.35/0.75, 6/0.57/1.78, 7/3.09/0.75, 8/3.84/1.78,
                        9/3.09/2.78, 10/1.35/2.78, 11/1.35/0.21, 12/3.09/0.21,
                        13/4.44/0.49, 14/4.44/2.19, 15/3.84/3.08, 16/2.20/3.59,
                        17/0.57/3.08, 18/0/2.19, 19/0/0.49, 20/2.20/4.14}
     \node[nd] (n\n) at (\x,\y) {\n};
  \foreach \a/\b in {1/2,1/3,1/4,2/5,2/7,3/6,3/10,4/8,4/9,5/6,5/11,6/18,
                     7/8,7/12,8/14,9/10,9/15,10/17,11/12,11/19,12/13,
                     13/14,14/15,15/16,16/17,16/20,17/18,18/19}
     \draw[ed] (n\a)--(n\b);
  \draw[ed] (n19) .. controls (-1.8,2.6) and (0.5,5.0) .. (n20);
  \draw[ed] (n13) .. controls (6.24,2.6) and (3.94,5.0) .. (n20);
  \node[st] at ([xshift=-1.0mm,yshift=3.0mm]n1) {s};
  \node[st] at ([xshift=1.4mm,yshift=3.0mm]n20) {t};
\end{tikzpicture}
\caption{$n=30$}\label{fig:tn19}%
\end{subfigure}\hfill
\begin{subfigure}[b]{0.49\linewidth}\centering
\begin{tikzpicture}[net,netb,scale=\netsc*1.65]
  \foreach \n/\x/\y in {1/0/3, 16/1/3, 8/2/3, 12/3/3,
                        2/0/2, 5/1/2, 9/2/2, 13/3/2,
                        3/0/1, 6/1/1, 10/2/1, 14/3/1,
                        4/0/0, 7/1/0, 11/2/0, 15/3/0}
     \node[nd] (n\n) at (\x,\y) {\n};
  \foreach \a/\b in {1/16,16/8,8/12, 2/5,5/9,9/13, 3/6,6/10,10/14, 4/7,7/11,11/15,
                     1/2,2/3,3/4, 16/5,5/6,6/7, 8/9,9/10,10/11, 12/13,13/14,14/15}
     \draw[ed] (n\a)--(n\b);
  \node[st] at ([xshift=-3.4mm,yshift=0mm]n1) {s};
  \node[st] at ([xshift=1.2mm,yshift=3.2mm]n16) {t};
\end{tikzpicture}
\caption{$n=24$}\label{fig:tn20}%
\end{subfigure}

\vspace{2ex}

\begin{subfigure}[b]{0.49\linewidth}\centering
\begin{tikzpicture}[net,netb,scale=\netsc*1.35]
  \foreach \n/\x/\y in {1/0/2.23, 2/1.20/4.09, 3/1.20/3.16, 4/1.20/2.23,
                        5/1.20/1.29, 6/1.20/0.36, 7/2.40/4.09, 8/3.50/4.09,
                        9/2.40/3.16, 10/2.40/2.23, 11/2.40/1.29, 12/2.40/0.36,
                        13/3.50/3.16, 14/4.48/0.36, 15/3.45/2.23, 16/4.48/3.16}
     \node[nd] (n\n) at (\x,\y) {\n};
  \foreach \a/\b in {1/2,1/3,1/4,1/5,1/6,2/3,2/7,3/4,3/7,3/8,3/9,
                     4/5,4/10,4/13,5/6,5/10,5/11,5/12,5/14,7/8,8/16,
                     9/13,10/11,10/13,12/14,13/14,13/15,13/16,14/15,14/16}
     \draw[ed] (n\a)--(n\b);
  \node[st] at ([xshift=-2.4mm,yshift=-3.2mm]n1) {s};%
  \node[st] at ([xshift=2.4mm,yshift=-3.2mm]n16) {t};%
\end{tikzpicture}
\caption{$n=30$}\label{fig:tn21}%
\end{subfigure}\hfill
\begin{subfigure}[b]{0.49\linewidth}\centering
\begin{tikzpicture}[net,netb,scale=\netsc*1.10]
  \foreach \n/\x/\y in {1/0/1.27, 2/1.09/2.16, 3/1.10/0.37, 4/2.24/1.27,
                        5/2.22/2.16, 6/2.24/0.37, 7/4.24/0.37, 8/3.24/0.37,
                        9/4.24/1.27, 10/3.26/1.27, 11/3.26/2.16, 12/4.26/2.16,
                        13/5.22/1.27}
     \node[nd] (n\n) at (\x,\y) {\n};
  \foreach \a/\b in {1/2,1/3,1/4,2/3,2/5,2/6,3/4,3/6,4/5,4/6,5/10,5/11,
                     6/8,7/8,7/9,7/13,8/9,9/10,9/12,10/11,11/12,12/13}
     \draw[ed] (n\a)--(n\b);
  \draw[ed] (n2) .. controls (-2.5,2.4) and (1.5,-2.6) .. (n7);
  \node[st] at ([xshift=-2.4mm,yshift=-3.2mm]n1) {s};%
  \node[st] at ([xshift=2.4mm,yshift=-3.2mm]n13) {t};%
\end{tikzpicture}
\caption{$n=23$}\label{fig:tn22}%
\end{subfigure}

\caption{The 22 test systems of Yeh et al.~\cite{yeh2024new}, networks 1--22 of their Figure~7. Each system is a two-terminal network whose components are the edges; it works iff there is a path of working edges from node s to node t. Below each network, $n$ is the number of components, i.e., of edges. In networks 6 and 7 every edge is an arc; in network 6 the four arcs leaving node 5 are irrelevant components. Network 21 is network 13 redrawn, and network 22 is network 18 plus one edge.}
\label{fig:yeh}
\end{figure}

\subsection{Cost structure}\label{subsec:costs}
We use the cost structure~\eqref{eq:cost setup} of Section~\ref{sec:policies}, with the cost of a system failure as the unit of cost, $\csys = 1$. We parametrize the costs by $\rho := n \ccomp / \csys$, the cost of replacing all components relative to that of a system failure, and by $\sigma := \cset / \csys$, the setup cost of a repair intervention, and we test $\rho \in \{1/10, \, 1/5, \, 1/2, \, 2/3, \, 1, \, 4/3, \, 2, \, 4, \, 10\}$ and $\sigma \in \{0, \, 0.05, \, 0.1, \, 0.2, \, 0.5\}$, i.e., $45$ cost pairs. Values $\rho < 1$ describe safety- or production-critical systems, in which a system failure is more expensive than a complete renewal, and $\rho > 1$ systems whose components are expensive compared with the consequences of a failure.

\subsection{Benchmark policies}\label{subsec:benchmarks}
We compare the best MTTF-threshold policy, with cost $v^\ast_\gamma := \min_\gamma LTMC^\gamma$, with two benchmarks. The first are the \emph{$r$-out-of-$n$ policies}, which repair as soon as $r$ components have failed or the system has failed, for $r = 1, \ldots, n$~\cite{finkelstein2015time,lagos2026simple}; we write $v^\ast_r$ for the cost of the best one. They only require counting failed components, and their cost follows from the signature of the complete system~\cite{lagos2026simple}. The one with $r = n$, i.e., corrective maintenance, coincides with $\pi^\gamma$ for $\gamma \to 0$, and, when all components are relevant, the one with $r = 1$ coincides with $\pi^1$. The second is the optimal state-based policy $\pi^\ast$ of Section~\ref{subsec:optimal state}, with cost $v^\ast_{\mathrm{state}}$. Since the other policies are state-based, $v^\ast_{\mathrm{state}} \leq \min\{v^\ast_\gamma, v^\ast_r\}$.
By Section~\ref{sec:dominance}, every signature-threshold policy is also state-based, so the gap $v^\ast_\gamma / v^\ast_{\mathrm{state}} - 1$ is an upper bound on the relative saving that any signature-threshold policy, whatever its weights, can achieve over the best MTTF-threshold policy.

\subsection{Experimental design and computational setup}\label{subsec:design}

\paragraph{Configurations}
A \emph{configuration} consists of a system (22 networks of Figure~\ref{fig:yeh}) and a failure model (31 models of Section~\ref{subsec:models exp}), which gives $682$ configurations, and $682 \times 45 = 30{,}690$ combinations of configuration and cost pair.
For each system we run Algorithm~\ref{alg} once, since it does not depend on the failure model, and store the classes of working states with the same counting vector, with their multiplicities. The same sweep also lists the minimal path sets (Section~\ref{sec:procedure}); they are not needed for the policies, and we use them only to check the implementation, against an independent count of the simple $s$--$t$ paths, and report their number in Table~\ref{tab:effort}. For each configuration we then compute, from the classes, the MTTF of all states, the whole curve $\gamma \mapsto (T^\gamma, N^\gamma, P^\gamma)$ at all its breakpoints (Proposition~\ref{prop:all thresholds}) and, for every cost pair, the best threshold and the best $r$-out-of-$n$ policy.
We compute the optimal state-based policy for the 19 systems with $n \leq 27$, i.e., for all systems except the three with $n = 30$, which gives $26{,}505$ combinations.

\paragraph{Implementation}
All computations are implemented in Julia~\cite{bezanson2017julia} and run sequentially, one system or configuration per core, on the high-performance computing cluster of Universidad Adolfo Ib\'a\~nez (nodes with two Intel Xeon E5-2650 processors at 2.0~GHz and 128~GB of RAM); all times below are measured on one core. Algorithm~\ref{alg} stores each level as an integer matrix with one column per state, ordered lexicographically so that the column of each son is computed directly, and evaluates $\Phi$ by a reachability search on the working edges. The values $\Psi(1), \ldots, \Psi(n)$ and $P_{w,w'}$ are computed in 512-bit arithmetic, since~\eqref{eq:P LFMO} is an alternating sum; the recursions given after Proposition~\ref{prop:all thresholds} involve only nonnegative terms, so double precision suffices for the rest. Dinkelbach's method is started at a value $g_0 \leq v^\ast_\gamma$ and stops when $h(\mathbf 1) \geq -10^{-12} \max\{1, c^{\pi_k}_{\mathrm{rep}}(\mathbf 1)\}$. The code, the experiment grid and the computed signatures are available at \url{https://github.com/guidolagos/Subsystems}.

\subsection{Results: the MTTF at all states}\label{subsec:results algorithm}
Table~\ref{tab:effort} reports, for each system, the computing times of Algorithm~\ref{alg} and of the enumeration of the $2^n$ states that is needed for the signature of the complete system alone.
Both evaluate the structure function $2^n$ times (Theorem~\ref{theo}), and Algorithm~\ref{alg} takes between $1.2$ and $5.6$ times as long as the enumeration for the systems with $n \geq 21$, at most about eight minutes for $n = 30$. Its peak memory, $12.4$~GB for systems 13 and 21, is dominated by the two largest consecutive levels of counting vectors, in agreement with the $\mathcal O(\sqrt n \, 2^n)$ bound of Theorem~\ref{theo}.

The MTTF can be computed state by state during the sweep, as in Algorithm~\ref{alg}: the column ``Alg.~\ref{alg} + MTTF'' of Table~\ref{tab:effort} does so for the 31 failure models at once, which makes the sweep between $1.6$ and $6.3$ times slower for the systems with $n \geq 21$, and up to $51$ minutes long for $n = 30$; moreover, the whole sweep would have to be repeated for each new failure model. This is unnecessary. By~\eqref{eq:E Samaniego sub}, the MTTF of a state depends only on its number of working components and its counting vector, so it suffices to compute it once for each \emph{class} of states with the same counting vector, and there are few such classes: for the systems with $n \geq 24$, between $0.005\%$ and $0.4\%$ of the $2^n$ states (column ``Classes''). We therefore run the sweep once per system, without any failure model, and store the classes with their multiplicities. Then, for each failure model, the MTTF of all states and the whole curve $\gamma \mapsto LTMC^\gamma$ are computed from the classes alone, in at most $2.8$ seconds, and in a median of less than $0.1$ seconds for every system with $n \leq 27$ (column ``Curve'').

\begin{table*}[!t]

\centering
\caption{
Computational effort for each system of Figure~\ref{fig:yeh} (one core): number $n$ of components; number of classes of working states with the same counting vector; number of minimal path sets (MPS); time to enumerate the $2^n$ states for the signature of the complete system; time of Algorithm~\ref{alg} (all counting vectors, classes and MPS); time of Algorithm~\ref{alg} computing also the MTTF of every state for the 31 failure models; peak memory of these runs; maximum number of breakpoints of $\gamma \mapsto LTMC^\gamma$ over the 31 models; time to compute that whole curve for one model from the classes (median over the models); and time to compute the optimal state-based policy for one cost pair (median over the models and cost pairs). Systems 6 and 7 are directed.
}
\label{tab:effort}
\small
\setlength{\tabcolsep}{4pt}
\begin{tabular}{@{}rrrrrrrrrrr@{}}
\toprule
& & & & \multicolumn{3}{c}{Time (s)} & & & \multicolumn{2}{c}{Time (s)} \\
\cmidrule(lr){5-7}\cmidrule(l){10-11}
System & $n$ & Classes & MPS & Enumeration & Alg.~\ref{alg} & Alg.~\ref{alg} + MTTF & Memory (GB) & Breakpoints & Curve & $\pi^\ast$ per cost pair \\
\midrule
1 & $5$ & $6$ & $4$ & $<0.01$ & $<0.01$ & $<0.01$ & $0.2$ & $5$ & $<0.01$ & $<0.01$ \\
2 & $8$ & $25$ & $7$ & $<0.01$ & $<0.01$ & $<0.01$ & $0.2$ & $18$ & $<0.01$ & $<0.01$ \\
3 & $8$ & $28$ & $9$ & $<0.01$ & $<0.01$ & $<0.01$ & $0.2$ & $21$ & $<0.01$ & $<0.01$ \\
4 & $9$ & $33$ & $13$ & $<0.01$ & $<0.01$ & $<0.01$ & $0.2$ & $26$ & $<0.01$ & $<0.01$ \\
5 & $12$ & $125$ & $13$ & $<0.01$ & $<0.01$ & $<0.01$ & $0.2$ & $56$ & $<0.01$ & $<0.01$ \\
6\,(D) & $14$ & $313$ & $7$ & $<0.01$ & $<0.01$ & $0.02$ & $0.2$ & $53$ & $<0.01$ & $<0.01$ \\
7\,(D) & $21$ & $12{,}940$ & $18$ & $0.13$ & $0.71$ & $3.4$ & $0.3$ & $4{,}194$ & $<0.01$ & $2.5$ \\
8 & $13$ & $294$ & $18$ & $<0.01$ & $<0.01$ & $<0.01$ & $0.2$ & $162$ & $<0.01$ & $<0.01$ \\
9 & $12$ & $207$ & $24$ & $<0.01$ & $<0.01$ & $<0.01$ & $0.2$ & $148$ & $<0.01$ & $<0.01$ \\
10 & $12$ & $223$ & $20$ & $<0.01$ & $<0.01$ & $<0.01$ & $0.2$ & $137$ & $<0.01$ & $<0.01$ \\
11 & $12$ & $326$ & $25$ & $<0.01$ & $<0.01$ & $<0.01$ & $0.2$ & $247$ & $<0.01$ & $<0.01$ \\
12 & $13$ & $244$ & $29$ & $<0.01$ & $<0.01$ & $<0.01$ & $0.2$ & $169$ & $<0.01$ & $<0.01$ \\
13 & $30$ & $4{,}549{,}544$ & $499$ & $117$ & $489$ & $3{,}051$ & $12.4$ & $1{,}309{,}230$ & $2.7$ & --- \\
14 & $26$ & $3{,}657$ & $44$ & $9.2$ & $11$ & $18$ & $1.0$ & $739$ & $<0.01$ & $36$ \\
15 & $14$ & $497$ & $44$ & $<0.01$ & $<0.01$ & $0.01$ & $0.2$ & $364$ & $<0.01$ & $<0.01$ \\
16 & $21$ & $40{,}880$ & $331$ & $0.15$ & $0.80$ & $4.4$ & $0.2$ & $33{,}549$ & $0.02$ & $2.4$ \\
17 & $27$ & $148{,}370$ & $269$ & $20$ & $37$ & $166$ & $1.8$ & $45{,}900$ & $0.06$ & $197$ \\
18 & $22$ & $65{,}639$ & $269$ & $0.47$ & $1.3$ & $5.7$ & $0.3$ & $41{,}122$ & $0.03$ & $4.4$ \\
19 & $30$ & $90{,}202$ & $780$ & $171$ & $337$ & $1{,}483$ & $9.5$ & $39{,}096$ & $0.05$ & --- \\
20 & $24$ & $9{,}015$ & $98$ & $1.7$ & $6.2$ & $37$ & $0.5$ & $4{,}273$ & $<0.01$ & $18$ \\
21 & $30$ & $4{,}549{,}544$ & $499$ & $110$ & $484$ & $3{,}032$ & $12.4$ & $1{,}309{,}230$ & $2.6$ & --- \\
22 & $23$ & $138{,}905$ & $441$ & $0.97$ & $2.5$ & $13$ & $0.4$ & $94{,}306$ & $0.08$ & $9.8$ \\
\bottomrule
\end{tabular}
\end{table*}

\subsection{Results: the MTTF-threshold policies}\label{subsec:results thresholds}

\paragraph{Number of policies}
The number of distinct policies $\pi^\gamma$, i.e., of breakpoints of $\gamma \mapsto LTMC^\gamma$, ranges from 5 for the bridge (system 1) to about $1.3$ million for systems 13 and 21 (Table~\ref{tab:effort}). It varies little across the failure models of a system.

\paragraph{Shape of the cost curve}
The curve $\gamma \mapsto LTMC^\gamma$ is neither convex nor unimodal in general, i.e., the cost can increase and then decrease again as $\gamma$ varies (Section~\ref{subsec:threshold}). Among the $20{,}925$ combinations of configuration and cost pair whose curve has at most $10{,}000$ breakpoints, $74\%$ have a breakpoint whose cost exceeds the best cost on both of its sides, by more than $0.1\%$ in $48\%$ of them and by more than $1\%$ in $5.1\%$, and by up to $5.3\%$. Figure~\ref{fig:curves} shows two examples. As a consequence, simple search strategies may miss the best threshold: decreasing $\gamma$ from $1$ until the cost first increases loses more than $1\%$ with respect to the best threshold in $23\%$ of these combinations, more than $5\%$ in $9.7\%$ of them, and up to $62\%$; a grid of 201 equally spaced thresholds loses up to $2.9\%$, and more than $1\%$ in $1.2\%$ of all $30{,}690$ combinations. Evaluating all breakpoints, which Proposition~\ref{prop:all thresholds} makes cheap, is therefore what guarantees the best threshold.

\paragraph{Best threshold}
The best threshold is interior in most cases: it is $\gamma = 1$, i.e., repairing at the first failure of a relevant component, in $2.6\%$ of the combinations, and corrective maintenance in $3.7\%$; its median is $\gamma = 0.69$. The best MTTF-threshold policy saves a median of $20\%$, and up to $97\%$, with respect to corrective maintenance.
It is never worse than the best $r$-out-of-$n$ policy, and it saves a median of $4.2\%$ with respect to it, with a 90th percentile of $13\%$ and a maximum of $54\%$. Thus, among simple policies, a threshold on the MTTF of the current state uses the information on \emph{which} components failed effectively.

\begin{figure}[!t]

\centering
\begin{tikzpicture}
\begin{groupplot}[
  group style={group size=1 by 2, vertical sep=1.5cm},
  width=\linewidth, height=0.55\linewidth,
  xmin=0, xmax=1, xlabel={$\gamma$}, ylabel={$LTMC^\gamma / v^\ast_{\mathrm{state}}$},
  grid=both, grid style={line width=0.3pt, draw=gray!25},
  tick label style={font=\footnotesize}, label style={font=\footnotesize},
  title style={font=\footnotesize},
  legend style={font=\scriptsize, draw=gray!50, at={(0.02,0.98)}, anchor=north west},
]
\nextgroupplot[title={(a) System 5 ($n=12$), [Exp], $\rho = 0.1$, $\sigma = 0$}, ymin=0.98, ymax=1.2]
\addplot[black, line width=0.8pt] coordinates {(0.0,20.39981) (0.26321,20.39981) (0.26321,20.2614) (0.30708,20.2614) (0.30708,19.33018) (0.33998,19.33018) (0.33998,19.16472) (0.36557,19.16472) (0.36557,19.11844) (0.36849,19.11844) (0.36849,17.55864) (0.38385,17.55864) (0.38385,17.49404) (0.4014,17.49404) (0.4014,17.42911) (0.41236,17.42911) (0.41236,17.06731) (0.44526,17.06731) (0.44526,16.79189) (0.46062,16.79189) (0.46062,14.38451) (0.47085,14.38451) (0.47085,14.23752) (0.47378,14.23752) (0.47378,14.08141) (0.49352,14.08141) (0.49352,13.91938) (0.50156,13.91938) (0.50156,13.86163) (0.50449,13.86163) (0.50449,13.48625) (0.51692,13.48625) (0.51692,13.41692) (0.52203,13.41692) (0.52203,12.9693) (0.53739,12.9693) (0.53739,12.88515) (0.56225,12.88515) (0.56225,12.85066) (0.5659,12.85066) (0.5659,12.57406) (0.58053,12.57406) (0.58053,12.48769) (0.58126,12.48769) (0.58126,12.49149) (0.59149,12.49149) (0.59149,12.40324) (0.59661,12.40324) (0.59661,12.32619) (0.60538,12.32619) (0.60538,12.06775) (0.61416,12.06775) (0.61416,7.96872) (0.6222,7.96872) (0.6222,7.95135) (0.63084,7.95135) (0.63084,7.59259) (0.6456,7.59259) (0.6456,7.54624) (0.65803,7.54624) (0.65803,6.70606) (0.6763,6.70606) (0.6763,6.50839) (0.69458,6.50839) (0.69458,6.44359) (0.70189,6.44359) (0.70189,6.4747) (0.70628,6.4747) (0.70628,5.0026) (0.71213,5.0026) (0.71213,4.90836) (0.71725,4.90836) (0.71725,4.29927) (0.71944,4.29927) (0.71944,3.95525) (0.73971,3.95525) (0.73971,3.68822) (0.7544,3.68822) (0.7544,3.78648) (0.756,3.78648) (0.756,3.31229) (0.75746,3.31229) (0.75746,3.00782) (0.76623,3.00782) (0.76623,2.84693) (0.78963,2.84693) (0.78963,2.44308) (0.80339,2.44308) (0.80339,2.0366) (0.81083,2.0366) (0.81083,1.00241) (0.81157,1.00241) (0.81157,1.00759) (0.82984,1.00759) (0.82984,1.01083) (0.87358,1.01083) (0.87358,1.01283) (0.87664,1.01283) (0.87664,1.02031) (0.8993,1.02031) (0.8993,1.03115) (0.93347,1.03115) (0.93347,1.03487) (0.94536,1.03487) (0.94536,1.04637) (0.95487,1.04637) (0.95487,1.06247) (1.0,1.06247)};
\addlegendentry{$LTMC^\gamma$}
\addplot[blue, dashed, line width=0.9pt, domain=0:1, samples=2] {1.0624669};
\addlegendentry{best $r$-out-of-$n$ ($r=1$)}
\addplot[red, line width=0.9pt, domain=0:1, samples=2] {1};
\addlegendentry{optimal state-based}
\addplot[only marks, mark=*, mark size=1.6pt, black, forget plot] coordinates {(0.81157,1.00241)};
\nextgroupplot[title={(b) System 14 ($n=26$), [Exp], $\rho = 0.1$, $\sigma = 0.05$}, ymin=0.9, ymax=2.6]
\addplot[black, line width=0.8pt] coordinates {(0.0,5.35975) (0.18234,5.35975) (0.18234,5.35966) (0.19306,5.35966) (0.19306,5.35898) (0.20266,5.35898) (0.20266,5.35874) (0.20513,5.35874) (0.20513,5.35804) (0.21339,5.35804) (0.21339,5.35782) (0.21473,5.35782) (0.21473,5.3576) (0.21586,5.3576) (0.21586,5.35709) (0.2188,5.35709) (0.2188,5.356) (0.22545,5.356) (0.22545,5.3556) (0.2284,5.3556) (0.2284,5.3554) (0.2285,5.3554) (0.2285,5.3546) (0.22953,5.3546) (0.22953,5.35413) (0.23443,5.35413) (0.23443,5.34982) (0.2416,5.34982) (0.2416,5.34939) (0.24465,5.34939) (0.24465,5.34866) (0.24516,5.34866) (0.24516,5.34721) (0.24725,5.34721) (0.24725,5.34648) (0.24811,5.34648) (0.24811,5.34388) (0.24924,5.34388) (0.24924,5.34224) (0.25247,5.34224) (0.25247,5.33357) (0.25267,5.33357) (0.25267,5.3335) (0.25723,5.3335) (0.25723,5.33312) (0.25979,5.33312) (0.25979,5.33247) (0.2614,5.33247) (0.2614,5.33181) (0.26243,5.33181) (0.26243,5.33034) (0.26414,5.33034) (0.26414,5.32693) (0.2681,5.32693) (0.2681,5.32474) (0.27279,5.32474) (0.27279,5.32424) (0.27351,5.32424) (0.27351,5.30087) (0.28613,5.30087) (0.28613,5.29893) (0.28913,5.29893) (0.28913,5.29699) (0.2902,5.29699) (0.2902,5.29677) (0.29154,5.29677) (0.29154,5.28962) (0.2963,5.28962) (0.2963,5.28915) (0.29837,5.28915) (0.29837,5.21481) (0.30717,5.21481) (0.30717,5.21139) (0.31032,5.21139) (0.31032,5.21119) (0.31258,5.21119) (0.31258,5.20893) (0.31348,5.20893) (0.31348,5.2029) (0.31371,5.2029) (0.31371,5.20271) (0.314,5.20271) (0.314,5.20099) (0.3164,5.20099) (0.3164,5.19468) (0.31642,5.19468) (0.31642,5.19424) (0.3166,5.19424) (0.3166,5.19419) (0.31941,5.19419) (0.31941,5.1708) (0.31957,5.1708) (0.31957,5.1707) (0.32411,5.1707) (0.32411,5.17003) (0.3256,5.17003) (0.3256,5.16987) (0.32767,5.16987) (0.32767,5.16925) (0.32821,5.16925) (0.32821,5.03182) (0.33099,5.03182) (0.33099,5.03173) (0.33383,5.03173) (0.33383,5.03139) (0.33744,5.03139) (0.33744,5.02078) (0.3378,5.02078) (0.3378,5.02067) (0.34095,5.02067) (0.34095,5.01977) (0.34152,5.01977) (0.34152,5.01961) (0.34375,5.01961) (0.34375,5.01427) (0.34423,5.01427) (0.34423,5.01389) (0.34427,5.01389) (0.34427,5.01007) (0.34443,5.01007) (0.34443,5.00999) (0.34526,5.00999) (0.34526,5.00333) (0.34571,5.00333) (0.34571,5.00293) (0.34624,5.00293) (0.34624,4.99062) (0.34853,4.99062) (0.34853,4.99039) (0.34925,4.99039) (0.34925,4.98654) (0.35047,4.98654) (0.35047,4.98629) (0.3506,4.98629) (0.3506,4.98092) (0.35307,4.98092) (0.35307,4.91814) (0.35395,4.91814) (0.35395,4.91782) (0.35548,4.91782) (0.35548,4.90534) (0.35751,4.90534) (0.35751,4.90328) (0.35777,4.90328) (0.35777,4.90221) (0.36187,4.90221) (0.36187,4.89836) (0.36236,4.89836) (0.36236,4.89819) (0.36467,4.89819) (0.36467,4.71917) (0.36574,4.71917) (0.36574,4.71873) (0.36728,4.71873) (0.36728,4.71554) (0.36774,4.71554) (0.36774,4.71522) (0.3685,4.71522) (0.3685,4.71511) (0.36884,4.71511) (0.36884,4.71497) (0.37306,4.71497) (0.37306,4.71021) (0.37411,4.71021) (0.37411,4.67731) (0.37486,4.67731) (0.37486,4.67652) (0.37518,4.67652) (0.37518,4.67638) (0.3754,4.67638) (0.3754,4.67556) (0.376,4.67556) (0.376,4.67075) (0.37654,4.67075) (0.37654,4.6705) (0.37763,4.6705) (0.37763,4.66937) (0.37794,4.66937) (0.37794,4.66828) (0.37892,4.66828) (0.37892,4.65726) (0.3803,4.65726) (0.3803,4.6568) (0.38255,4.6568) (0.38255,4.65602) (0.38281,4.65602) (0.38281,4.63328) (0.38291,4.63328) (0.38291,4.60498) (0.38395,4.60498) (0.38395,4.59365) (0.38485,4.59365) (0.38485,4.59322) (0.385,4.59322) (0.385,4.59314) (0.38533,4.59314) (0.38533,4.59118) (0.38571,4.59118) (0.38571,4.57657) (0.38723,4.57657) (0.38723,4.57642) (0.38747,4.57642) (0.38747,4.57611) (0.38804,4.57611) (0.38804,4.57434) (0.38954,4.57434) (0.38954,4.53972) (0.3903,4.53972) (0.3903,4.53796) (0.39134,4.53796) (0.39134,4.53781) (0.39161,4.53781) (0.39161,4.53628) (0.39214,4.53628) (0.39214,4.52881) (0.39451,4.52881) (0.39451,4.45344) (0.39462,4.45344) (0.39462,4.45306) (0.39483,4.45306) (0.39483,4.45125) (0.39495,4.45125) (0.39495,4.43466) (0.39515,4.43466) (0.39515,4.418) (0.39658,4.418) (0.39658,4.41697) (0.39834,4.41697) (0.39834,4.41579) (0.39914,4.41579) (0.39914,4.41577) (0.40114,4.41577) (0.40114,4.41535) (0.40171,4.41535) (0.40171,4.41496) (0.40221,4.41496) (0.40221,4.4148) (0.40321,4.4148) (0.40321,4.41429) (0.40375,4.41429) (0.40375,4.40448) (0.40438,4.40448) (0.40438,4.39093) (0.4062,4.39093) (0.4062,4.39057) (0.4071,4.39057) (0.4071,4.39045) (0.40765,4.39045) (0.40765,4.3898) (0.40831,4.3898) (0.40831,4.38922) (0.40865,4.38922) (0.40865,4.38876) (0.41006,4.38876) (0.41006,4.38809) (0.41021,4.38809) (0.41021,4.38791) (0.41026,4.38791) (0.41026,4.25005) (0.41058,4.25005) (0.41058,4.24024) (0.41074,4.24024) (0.41074,4.24011) (0.41097,4.24011) (0.41097,4.23943) (0.41167,4.23943) (0.41167,4.23751) (0.4129,4.23751) (0.4129,4.23738) (0.41374,4.23738) (0.41374,4.23632) (0.41411,4.23632) (0.41411,4.23619) (0.41426,4.23619) (0.41426,4.2347) (0.41677,4.2347) (0.41677,4.2343) (0.41702,4.2343) (0.41702,4.23396) (0.41724,4.23396) (0.41724,4.23355) (0.4181,4.23355) (0.4181,4.23327) (0.41864,4.23327) (0.41864,4.22814) (0.41938,4.22814) (0.41938,4.19738) (0.41939,4.19738) (0.41939,4.1971) (0.41946,4.1971) (0.41946,4.19681) (0.42145,4.19681) (0.42145,4.1966) (0.42198,4.1966) (0.42198,4.19106) (0.42233,4.19106) (0.42233,4.19023) (0.42296,4.19023) (0.42296,4.18945) (0.42365,4.18945) (0.42365,4.18911) (0.42408,4.18911) (0.42408,4.1884) (0.42437,4.1884) (0.42437,4.18312) (0.42449,4.18312) (0.42449,4.18161) (0.42539,4.18161) (0.42539,4.14832) (0.42661,4.14832) (0.42661,4.14761) (0.42834,4.14761) (0.42834,4.14691) (0.42861,4.14691) (0.42861,4.13321) (0.42876,4.13321) (0.42876,4.13241) (0.42921,4.13241) (0.42921,4.12194) (0.42937,4.12194) (0.42937,4.12162) (0.43098,4.12162) (0.43098,4.11118) (0.43162,4.11118) (0.43162,4.09636) (0.432,4.09636) (0.432,4.01684) (0.43217,4.01684) (0.43217,4.01525) (0.43305,4.01525) (0.43305,4.01449) (0.43351,4.01449) (0.43351,4.01438) (0.43512,4.01438) (0.43512,3.98674) (0.43612,3.98674) (0.43612,3.98642) (0.43686,3.98642) (0.43686,3.98568) (0.43719,3.98568) (0.43719,3.98534) (0.43761,3.98534) (0.43761,3.98217) (0.43794,3.98217) (0.43794,3.9765) (0.4401,3.9765) (0.4401,3.96558) (0.44041,3.96558) (0.44041,3.9603) (0.44093,3.9603) (0.44093,3.88795) (0.44109,3.88795) (0.44109,3.88743) (0.44195,3.88743) (0.44195,3.88719) (0.44195,3.88719) (0.44195,3.87592) (0.44294,3.87592) (0.44294,3.86661) (0.44306,3.86661) (0.44306,3.86583) (0.44379,3.86583) (0.44379,3.86551) (0.44409,3.86551) (0.44409,3.86495) (0.44416,3.86495) (0.44416,3.86112) (0.44424,3.86112) (0.44424,3.85959) (0.44574,3.85959) (0.44574,3.85367) (0.44604,3.85367) (0.44604,3.85356) (0.44673,3.85356) (0.44673,3.77376) (0.44738,3.77376) (0.44738,3.76982) (0.44875,3.76982) (0.44875,3.7697) (0.449,3.7697) (0.449,3.76946) (0.44943,3.76946) (0.44943,3.76921) (0.45055,3.76921) (0.45055,3.74177) (0.45069,3.74177) (0.45069,3.74096) (0.45094,3.74096) (0.45094,3.74004) (0.45125,3.74004) (0.45125,3.73805) (0.45296,3.73805) (0.45296,3.73724) (0.45312,3.73724) (0.45312,3.7369) (0.45405,3.7369) (0.45405,3.72839) (0.45513,3.72839) (0.45513,3.72758) (0.45533,3.72758) (0.45533,3.72733) (0.45583,3.72733) (0.45583,3.72681) (0.45584,3.72681) (0.45584,3.72679) (0.45616,3.72679) (0.45616,3.71596) (0.45701,3.71596) (0.45701,3.71377) (0.45792,3.71377) (0.45792,3.71317) (0.45803,3.71317) (0.45803,3.71025) (0.45845,3.71025) (0.45845,3.70829) (0.45887,3.70829) (0.45887,3.69969) (0.45944,3.69969) (0.45944,3.69876) (0.45947,3.69876) (0.45947,3.69816) (0.45992,3.69816) (0.45992,3.69738) (0.45995,3.69738) (0.45995,3.69712) (0.46001,3.69712) (0.46001,3.697) (0.46114,3.697) (0.46114,3.69614) (0.46115,3.69614) (0.46115,3.69324) (0.46335,3.69324) (0.46335,3.69263) (0.46344,3.69263) (0.46344,3.69102) (0.46381,3.69102) (0.46381,3.68971) (0.46466,3.68971) (0.46466,3.68968) (0.46496,3.68968) (0.46496,3.68268) (0.46581,3.68268) (0.46581,3.67881) (0.46637,3.67881) (0.46637,3.67751) (0.46654,3.67751) (0.46654,3.67458) (0.46764,3.67458) (0.46764,3.67326) (0.46783,3.67326) (0.46783,3.67266) (0.46887,3.67266) (0.46887,3.58116) (0.46893,3.58116) (0.46893,3.57973) (0.46952,3.57973) (0.46952,3.57946) (0.47052,3.57946) (0.47052,3.57884) (0.47131,3.57884) (0.47131,3.57855) (0.47194,3.57855) (0.47194,3.57827) (0.47264,3.57827) (0.47264,3.57658) (0.47408,3.57658) (0.47408,3.5756) (0.4746,3.5756) (0.4746,3.56582) (0.47527,3.56582) (0.47527,3.56427) (0.47641,3.56427) (0.47641,3.56362) (0.47656,3.56362) (0.47656,3.54406) (0.47688,3.54406) (0.47688,3.54125) (0.47766,3.54125) (0.47766,3.53646) (0.47862,3.53646) (0.47862,3.52654) (0.47864,3.52654) (0.47864,3.52587) (0.47949,3.52587) (0.47949,3.52193) (0.4797,3.52193) (0.4797,3.52185) (0.47982,3.52185) (0.47982,3.52013) (0.48071,3.52013) (0.48071,3.52015) (0.48101,3.52015) (0.48101,3.5134) (0.48151,3.5134) (0.48151,3.51342) (0.48189,3.51342) (0.48189,3.50614) (0.48206,3.50614) (0.48206,3.49934) (0.4821,3.49934) (0.4821,3.49921) (0.48254,3.49921) (0.48254,3.49789) (0.48421,3.49789) (0.48421,3.47456) (0.48482,3.47456) (0.48482,3.47132) (0.486,3.47132) (0.486,3.4594) (0.48684,3.4594) (0.48684,3.45707) (0.4869,3.45707) (0.4869,3.45541) (0.48775,3.45541) (0.48775,3.45501) (0.4882,3.45501) (0.4882,3.45175) (0.49036,3.45175) (0.49036,3.44997) (0.49066,3.44997) (0.49066,3.40512) (0.49072,3.40512) (0.49072,3.40456) (0.49122,3.40456) (0.49122,3.40442) (0.4915,3.40442) (0.4915,3.40203) (0.49235,3.40203) (0.49235,3.40054) (0.49314,3.40054) (0.49314,3.36908) (0.49316,3.36908) (0.49316,3.36659) (0.49381,3.36659) (0.49381,3.3659) (0.49492,3.3659) (0.49492,3.36202) (0.49535,3.36202) (0.49535,3.35779) (0.49651,3.35779) (0.49651,3.35709) (0.49664,3.35709) (0.49664,3.35639) (0.49804,3.35639) (0.49804,3.34789) (0.49807,3.34789) (0.49807,3.34755) (0.4987,3.34755) (0.4987,3.33696) (0.49982,3.33696) (0.49982,3.33698) (0.50092,3.33698) (0.50092,3.32081) (0.50133,3.32081) (0.50133,3.30452) (0.50213,3.30452) (0.50213,3.30405) (0.50216,3.30405) (0.50216,3.29663) (0.50306,3.29663) (0.50306,3.29407) (0.50333,3.29407) (0.50333,3.28971) (0.50372,3.28971) (0.50372,3.28955) (0.50403,3.28955) (0.50403,3.28527) (0.50449,3.28527) (0.50449,3.2848) (0.50529,3.2848) (0.50529,3.28441) (0.50533,3.28441) (0.50533,3.25785) (0.50662,3.25785) (0.50662,3.25695) (0.50704,3.25695) (0.50704,3.25251) (0.50711,3.25251) (0.50711,3.25229) (0.50748,3.25229) (0.50748,3.24144) (0.50774,3.24144) (0.50774,3.23984) (0.50809,3.23984) (0.50809,3.23968) (0.50969,3.23968) (0.50969,3.2342) (0.50997,3.2342) (0.50997,3.23237) (0.51008,3.23237) (0.51008,3.22687) (0.51054,3.22687) (0.51054,3.22663) (0.51069,3.22663) (0.51069,3.22405) (0.51151,3.22405) (0.51151,3.22265) (0.51171,3.22265) (0.51171,3.22135) (0.51279,3.22135) (0.51279,3.21774) (0.51332,3.21774) (0.51332,3.16457) (0.51334,3.16457) (0.51334,3.16425) (0.51367,3.16425) (0.51367,3.15494) (0.51395,3.15494) (0.51395,3.15494) (0.5145,3.15494) (0.5145,3.15476) (0.51477,3.15476) (0.51477,3.15143) (0.51522,3.15143) (0.51522,3.15102) (0.51557,3.15102) (0.51557,3.15102) (0.51671,3.15102) (0.51671,3.15062) (0.51984,3.15062) (0.51984,3.15018) (0.52014,3.15018) (0.52014,3.14947) (0.52127,3.14947) (0.52127,3.14848) (0.52207,3.14848) (0.52207,3.11608) (0.52241,3.11608) (0.52241,3.11566) (0.52247,3.11566) (0.52247,3.11338) (0.5232,3.11338) (0.5232,3.1119) (0.52357,3.1119) (0.52357,3.11137) (0.52386,3.11137) (0.52386,3.10183) (0.52422,3.10183) (0.52422,3.10182) (0.52446,3.10182) (0.52446,3.10101) (0.52564,3.10101) (0.52564,3.10089) (0.52627,3.10089) (0.52627,3.10047) (0.52646,3.10047) (0.52646,3.09966) (0.52661,3.09966) (0.52661,3.09974) (0.52749,3.09974) (0.52749,3.09973) (0.52832,3.09973) (0.52832,3.09877) (0.529,3.09877) (0.529,3.09232) (0.53137,3.09232) (0.53137,3.09231) (0.5317,3.09231) (0.5317,3.08564) (0.53173,3.08564) (0.53173,3.08563) (0.53223,3.08563) (0.53223,3.08421) (0.53258,3.08421) (0.53258,3.08317) (0.5362,3.08317) (0.5362,3.08167) (0.53629,3.08167) (0.53629,3.08178) (0.53708,3.08178) (0.53708,3.08185) (0.53724,3.08185) (0.53724,3.08162) (0.53801,3.08162) (0.53801,3.07679) (0.5405,3.07679) (0.5405,3.07573) (0.54091,3.07573) (0.54091,3.07333) (0.54119,3.07333) (0.54119,3.07234) (0.54131,3.07234) (0.54131,3.07135) (0.54454,3.07135) (0.54454,3.07107) (0.54492,3.07107) (0.54492,3.07124) (0.5453,3.07124) (0.5453,3.07025) (0.5454,3.07025) (0.5454,3.07051) (0.54576,3.07051) (0.54576,3.0705) (0.54701,3.0705) (0.54701,2.75342) (0.54843,2.75342) (0.54843,2.75216) (0.54876,2.75216) (0.54876,2.75189) (0.55092,2.75189) (0.55092,2.7272) (0.55127,2.7272) (0.55127,2.72706) (0.55145,2.72706) (0.55145,2.71932) (0.55169,2.71932) (0.55169,2.71596) (0.55254,2.71596) (0.55254,2.7156) (0.5527,2.7156) (0.5527,2.70667) (0.55274,2.70667) (0.55274,2.70277) (0.55302,2.70277) (0.55302,2.70246) (0.55494,2.70246) (0.55494,2.7021) (0.55568,2.7021) (0.55568,2.70154) (0.55575,2.70154) (0.55575,2.69762) (0.55613,2.69762) (0.55613,2.69735) (0.55642,2.69735) (0.55642,2.68946) (0.55679,2.68946) (0.55679,2.68656) (0.55778,2.68656) (0.55778,2.68533) (0.55828,2.68533) (0.55828,2.68521) (0.56004,2.68521) (0.56004,2.67777) (0.56211,2.67777) (0.56211,2.67767) (0.5626,2.67767) (0.5626,2.67761) (0.56264,2.67761) (0.56264,2.67281) (0.56316,2.67281) (0.56316,2.66711) (0.56348,2.66711) (0.56348,2.66725) (0.56393,2.66725) (0.56393,2.6671) (0.56545,2.6671) (0.56545,2.66673) (0.56682,2.66673) (0.56682,2.66686) (0.56706,2.66686) (0.56706,2.66625) (0.5675,2.66625) (0.5675,2.63427) (0.5676,2.63427) (0.5676,2.63389) (0.56802,2.63389) (0.56802,2.63373) (0.56805,2.63373) (0.56805,2.62366) (0.56817,2.62366) (0.56817,2.62349) (0.5689,2.62349) (0.5689,2.6213) (0.56911,2.6213) (0.56911,2.62123) (0.56934,2.62123) (0.56934,2.62107) (0.57012,2.62107) (0.57012,2.60276) (0.57075,2.60276) (0.57075,2.60249) (0.57152,2.60249) (0.57152,2.60231) (0.57337,2.60231) (0.57337,2.60126) (0.57346,2.60126) (0.57346,2.59537) (0.57483,2.59537) (0.57483,2.59532) (0.57486,2.59532) (0.57486,2.59447) (0.57501,2.59447) (0.57501,2.59429) (0.5754,2.59429) (0.5754,2.58594) (0.57621,2.58594) (0.57621,2.58565) (0.57691,2.58565) (0.57691,2.58479) (0.57697,2.58479) (0.57697,2.57285) (0.57719,2.57285) (0.57719,2.55405) (0.57812,2.55405) (0.57812,2.55363) (0.57858,2.55363) (0.57858,2.55318) (0.57895,2.55318) (0.57895,2.55218) (0.57953,2.55218) (0.57953,2.55155) (0.57978,2.55155) (0.57978,2.55168) (0.58035,2.55168) (0.58035,2.55141) (0.58106,2.55141) (0.58106,2.55078) (0.58187,2.55078) (0.58187,2.55097) (0.58211,2.55097) (0.58211,2.5511) (0.58219,2.5511) (0.58219,2.55124) (0.58313,2.55124) (0.58313,2.55106) (0.58348,2.55106) (0.58348,2.45694) (0.58408,2.45694) (0.58408,2.45535) (0.5843,2.45535) (0.5843,2.45525) (0.5855,2.45525) (0.5855,2.45275) (0.58597,2.45275) (0.58597,2.45297) (0.58625,2.45297) (0.58625,2.45324) (0.58786,2.45324) (0.58786,2.45274) (0.58834,2.45274) (0.58834,2.45164) (0.58889,2.45164) (0.58889,2.43625) (0.59014,2.43625) (0.59014,2.43039) (0.59049,2.43039) (0.59049,2.37021) (0.59221,2.37021) (0.59221,2.36996) (0.59292,2.36996) (0.59292,2.3666) (0.59313,2.3666) (0.59313,2.36604) (0.59321,2.36604) (0.59321,2.36546) (0.59508,2.36546) (0.59508,2.3652) (0.59508,2.3652) (0.59508,2.32848) (0.59552,2.32848) (0.59552,2.32804) (0.5966,2.32804) (0.5966,2.32787) (0.59679,2.32787) (0.59679,2.3281) (0.59714,2.3281) (0.59714,2.32841) (0.59724,2.32841) (0.59724,2.32705) (0.5975,2.32705) (0.5975,2.32682) (0.5985,2.32682) (0.5985,2.32621) (0.59911,2.32621) (0.59911,2.3258) (0.59953,2.3258) (0.59953,2.32258) (0.59964,2.32258) (0.59964,2.32234) (0.6001,2.32234) (0.6001,2.32253) (0.60024,2.32253) (0.60024,2.31998) (0.6007,2.31998) (0.6007,2.31862) (0.60185,2.31862) (0.60185,2.31855) (0.60204,2.31855) (0.60204,2.31928) (0.60234,2.31928) (0.60234,2.30574) (0.60236,2.30574) (0.60236,2.30565) (0.60291,2.30565) (0.60291,2.30553) (0.60405,2.30553) (0.60405,2.30541) (0.6042,2.30541) (0.6042,2.30466) (0.60562,2.30466) (0.60562,2.10928) (0.60569,2.10928) (0.60569,2.10919) (0.60574,2.10919) (0.60574,2.10563) (0.60608,2.10563) (0.60608,2.1043) (0.60782,2.1043) (0.60782,2.10445) (0.60847,2.10445) (0.60847,2.10461) (0.60933,2.10461) (0.60933,2.09877) (0.61037,2.09877) (0.61037,2.09908) (0.61095,2.09908) (0.61095,2.09718) (0.61215,2.09718) (0.61215,2.09526) (0.61225,2.09526) (0.61225,1.99823) (0.61432,1.99823) (0.61432,1.99165) (0.61474,1.99165) (0.61474,1.99181) (0.61585,1.99181) (0.61585,1.9913) (0.6172,1.9913) (0.6172,1.99043) (0.61805,1.99043) (0.61805,1.99018) (0.61916,1.99018) (0.61916,1.98452) (0.62219,1.98452) (0.62219,1.98465) (0.62286,1.98465) (0.62286,1.97954) (0.62387,1.97954) (0.62387,1.97966) (0.6247,1.97966) (0.6247,1.97767) (0.62552,1.97767) (0.62552,1.97567) (0.62607,1.97567) (0.62607,1.97524) (0.62635,1.97524) (0.62635,1.97553) (0.62658,1.97553) (0.62658,1.97478) (0.62672,1.97478) (0.62672,1.97543) (0.62715,1.97543) (0.62715,1.97503) (0.62723,1.97503) (0.62723,1.97498) (0.628,1.97498) (0.628,1.9751) (0.62906,1.9751) (0.62906,1.94355) (0.62966,1.94355) (0.62966,1.93901) (0.63023,1.93901) (0.63023,1.93856) (0.63152,1.93856) (0.63152,1.93885) (0.6323,1.93885) (0.6323,1.93824) (0.63275,1.93824) (0.63275,1.93847) (0.63309,1.93847) (0.63309,1.9375) (0.63331,1.9375) (0.63331,1.93658) (0.63576,1.93658) (0.63576,1.93669) (0.63581,1.93669) (0.63581,1.93672) (0.63626,1.93672) (0.63626,1.93575) (0.6365,1.93575) (0.6365,1.92228) (0.63774,1.92228) (0.63774,1.9225) (0.6394,1.9225) (0.6394,1.92039) (0.64079,1.92039) (0.64079,1.91622) (0.6413,1.91622) (0.6413,1.91029) (0.64143,1.91029) (0.64143,1.90815) (0.64183,1.90815) (0.64183,1.90844) (0.64206,1.90844) (0.64206,1.90816) (0.64276,1.90816) (0.64276,1.90787) (0.6438,1.90787) (0.6438,1.90853) (0.64385,1.90853) (0.64385,1.90863) (0.64492,1.90863) (0.64492,1.90652) (0.64508,1.90652) (0.64508,1.87979) (0.64676,1.87979) (0.64676,1.87988) (0.64706,1.87988) (0.64706,1.84602) (0.64707,1.84602) (0.64707,1.84572) (0.64716,1.84572) (0.64716,1.8451) (0.65045,1.8451) (0.65045,1.84526) (0.65194,1.84526) (0.65194,1.84541) (0.6521,1.84541) (0.6521,1.84543) (0.65227,1.84543) (0.65227,1.81096) (0.65327,1.81096) (0.65327,1.79838) (0.65434,1.79838) (0.65434,1.79892) (0.65443,1.79892) (0.65443,1.80032) (0.65535,1.80032) (0.65535,1.80045) (0.65641,1.80045) (0.65641,1.80064) (0.65649,1.80064) (0.65649,1.79994) (0.65775,1.79994) (0.65775,1.80007) (0.65838,1.80007) (0.65838,1.79992) (0.6589,1.79992) (0.6589,1.77319) (0.65909,1.77319) (0.65909,1.77345) (0.65922,1.77345) (0.65922,1.77228) (0.65954,1.77228) (0.65954,1.77111) (0.66039,1.77111) (0.66039,1.77078) (0.66067,1.77078) (0.66067,1.77099) (0.66145,1.77099) (0.66145,1.77102) (0.6628,1.77102) (0.6628,1.75803) (0.66313,1.75803) (0.66313,1.75584) (0.66319,1.75584) (0.66319,1.75569) (0.66319,1.75569) (0.66319,1.75578) (0.66389,1.75578) (0.66389,1.7558) (0.66416,1.7558) (0.66416,1.75535) (0.66447,1.75535) (0.66447,1.66929) (0.66532,1.66929) (0.66532,1.6693) (0.66556,1.6693) (0.66556,1.66769) (0.66618,1.66769) (0.66618,1.66639) (0.6666,1.66639) (0.6666,1.66641) (0.66825,1.66641) (0.66825,1.66556) (0.6691,1.66556) (0.6691,1.66618) (0.6698,1.66618) (0.6698,1.66753) (0.67098,1.66753) (0.67098,1.66592) (0.67105,1.66592) (0.67105,1.66457) (0.67308,1.66457) (0.67308,1.6646) (0.6734,1.6646) (0.6734,1.63396) (0.67343,1.63396) (0.67343,1.63358) (0.6755,1.63358) (0.6755,1.62972) (0.6762,1.62972) (0.6762,1.46509) (0.67692,1.46509) (0.67692,1.46354) (0.67695,1.46354) (0.67695,1.46352) (0.67725,1.46352) (0.67725,1.46342) (0.67779,1.46342) (0.67779,1.46017) (0.67874,1.46017) (0.67874,1.46073) (0.67921,1.46073) (0.67921,1.46068) (0.67994,1.46068) (0.67994,1.46307) (0.68002,1.46307) (0.68002,1.46292) (0.68018,1.46292) (0.68018,1.45912) (0.68338,1.45912) (0.68338,1.45511) (0.68366,1.45511) (0.68366,1.45566) (0.68437,1.45566) (0.68437,1.44973) (0.68456,1.44973) (0.68456,1.44945) (0.68537,1.44945) (0.68537,1.44929) (0.68586,1.44929) (0.68586,1.44864) (0.68594,1.44864) (0.68594,1.34681) (0.68736,1.34681) (0.68736,1.34674) (0.68777,1.34674) (0.68777,1.35063) (0.68911,1.35063) (0.68911,1.34701) (0.69004,1.34701) (0.69004,1.34819) (0.69199,1.34819) (0.69199,1.34801) (0.69273,1.34801) (0.69273,1.34774) (0.69288,1.34774) (0.69288,1.34842) (0.69416,1.34842) (0.69416,1.34817) (0.6942,1.34817) (0.6942,1.34787) (0.69539,1.34787) (0.69539,1.33338) (0.69843,1.33338) (0.69843,1.33307) (0.69857,1.33307) (0.69857,1.33131) (0.69976,1.33131) (0.69976,1.33071) (0.70098,1.33071) (0.70098,1.33005) (0.70336,1.33005) (0.70336,1.32633) (0.70368,1.32633) (0.70368,1.32687) (0.7046,1.32687) (0.7046,1.3155) (0.70521,1.3155) (0.70521,1.3194) (0.7081,1.3194) (0.7081,1.31907) (0.7084,1.31907) (0.7084,1.31872) (0.70862,1.31872) (0.70862,1.31691) (0.70951,1.31691) (0.70951,1.3166) (0.71016,1.3166) (0.71016,1.31639) (0.71058,1.31639) (0.71058,1.31521) (0.7127,1.31521) (0.7127,1.31075) (0.71308,1.31075) (0.71308,1.3109) (0.71405,1.3109) (0.71405,1.30707) (0.71487,1.30707) (0.71487,1.3024) (0.71534,1.3024) (0.71534,1.30205) (0.7184,1.30205) (0.7184,1.30195) (0.71849,1.30195) (0.71849,1.30118) (0.71948,1.30118) (0.71948,1.29417) (0.7198,1.29417) (0.7198,1.29395) (0.72158,1.29395) (0.72158,1.29157) (0.72512,1.29157) (0.72512,1.29135) (0.72544,1.29135) (0.72544,1.29099) (0.72573,1.29099) (0.72573,1.29028) (0.72692,1.29028) (0.72692,1.2929) (0.72782,1.2929) (0.72782,1.29396) (0.72793,1.29396) (0.72793,1.29794) (0.72837,1.29794) (0.72837,1.29758) (0.72962,1.29758) (0.72962,1.29721) (0.72969,1.29721) (0.72969,1.29323) (0.73001,1.29323) (0.73001,1.29443) (0.73104,1.29443) (0.73104,1.29564) (0.73234,1.29564) (0.73234,1.29967) (0.73263,1.29967) (0.73263,1.28703) (0.73294,1.28703) (0.73294,1.2667) (0.7334,1.2667) (0.7334,1.26531) (0.73444,1.26531) (0.73444,1.26441) (0.73698,1.26441) (0.73698,1.26416) (0.73717,1.26416) (0.73717,1.26314) (0.73912,1.26314) (0.73912,1.26368) (0.73934,1.26368) (0.73934,1.26357) (0.74065,1.26357) (0.74065,1.25321) (0.74144,1.25321) (0.74144,1.23615) (0.74199,1.23615) (0.74199,1.2357) (0.74244,1.2357) (0.74244,1.23687) (0.74289,1.23687) (0.74289,1.23591) (0.74335,1.23591) (0.74335,1.23495) (0.74405,1.23495) (0.74405,1.23529) (0.74411,1.23529) (0.74411,1.23501) (0.74624,1.23501) (0.74624,1.23608) (0.7466,1.23608) (0.7466,1.23564) (0.7479,1.23564) (0.7479,1.23376) (0.752,1.23376) (0.752,1.22938) (0.75268,1.22938) (0.75268,1.2291) (0.75344,1.2291) (0.75344,1.22881) (0.75398,1.22881) (0.75398,1.23126) (0.75537,1.23126) (0.75537,1.23113) (0.75607,1.23113) (0.75607,1.23017) (0.75682,1.23017) (0.75682,1.23887) (0.75764,1.23887) (0.75764,1.23818) (0.75767,1.23818) (0.75767,1.23874) (0.7581,1.23874) (0.7581,1.23664) (0.7584,1.23664) (0.7584,1.23619) (0.76032,1.23619) (0.76032,1.23744) (0.76065,1.23744) (0.76065,1.24159) (0.76088,1.24159) (0.76088,1.24001) (0.76146,1.24001) (0.76146,1.24901) (0.76199,1.24901) (0.76199,1.24819) (0.7629,1.24819) (0.7629,1.24806) (0.76325,1.24806) (0.76325,1.23375) (0.76527,1.23375) (0.76527,1.23274) (0.76564,1.23274) (0.76564,1.23171) (0.76686,1.23171) (0.76686,1.23299) (0.76746,1.23299) (0.76746,1.23333) (0.76935,1.23333) (0.76935,1.23447) (0.77579,1.23447) (0.77579,1.23713) (0.77581,1.23713) (0.77581,1.23728) (0.77825,1.23728) (0.77825,1.2368) (0.77843,1.2368) (0.77843,1.23581) (0.77917,1.23581) (0.77917,1.23832) (0.7804,1.23832) (0.7804,1.24389) (0.78095,1.24389) (0.78095,1.24679) (0.78175,1.24679) (0.78175,1.24577) (0.78314,1.24577) (0.78314,1.29812) (0.78501,1.29812) (0.78501,1.2965) (0.78524,1.2965) (0.78524,1.2962) (0.78619,1.2962) (0.78619,1.29525) (0.78647,1.29525) (0.78647,1.29323) (0.78741,1.29323) (0.78741,1.2947) (0.78929,1.2947) (0.78929,1.29559) (0.79002,1.29559) (0.79002,1.29547) (0.7903,1.29547) (0.7903,1.29451) (0.79173,1.29451) (0.79173,1.29284) (0.79409,1.29284) (0.79409,1.29191) (0.79486,1.29191) (0.79486,1.3025) (0.79649,1.3025) (0.79649,1.30097) (0.79761,1.30097) (0.79761,1.30004) (0.79835,1.30004) (0.79835,1.29935) (0.79856,1.29935) (0.79856,1.30088) (0.80021,1.30088) (0.80021,1.30236) (0.80324,1.30236) (0.80324,1.30532) (0.80591,1.30532) (0.80591,1.30326) (0.80671,1.30326) (0.80671,1.29677) (0.80733,1.29677) (0.80733,1.29287) (0.80815,1.29287) (0.80815,1.29184) (0.81098,1.29184) (0.81098,1.29859) (0.81126,1.29859) (0.81126,1.29811) (0.81133,1.29811) (0.81133,1.30163) (0.81285,1.30163) (0.81285,1.30324) (0.81452,1.30324) (0.81452,1.30107) (0.81618,1.30107) (0.81618,1.30008) (0.82101,1.30008) (0.82101,1.30547) (0.82241,1.30547) (0.82241,1.2932) (0.82369,1.2932) (0.82369,1.29367) (0.82661,1.29367) (0.82661,1.29257) (0.82675,1.29257) (0.82675,1.28917) (0.82858,1.28917) (0.82858,1.2902) (0.82959,1.2902) (0.82959,1.29122) (0.82985,1.29122) (0.82985,1.28473) (0.83042,1.28473) (0.83042,1.28702) (0.83051,1.28702) (0.83051,1.2875) (0.83165,1.2875) (0.83165,1.28464) (0.83201,1.28464) (0.83201,1.29705) (0.83249,1.29705) (0.83249,1.32141) (0.83266,1.32141) (0.83266,1.31411) (0.83461,1.31411) (0.83461,1.31297) (0.83615,1.31297) (0.83615,1.33515) (0.84024,1.33515) (0.84024,1.33634) (0.84437,1.33634) (0.84437,1.34039) (0.84451,1.34039) (0.84451,1.34695) (0.84651,1.34695) (0.84651,1.34816) (0.84775,1.34816) (0.84775,1.34489) (0.85395,1.34489) (0.85395,1.3593) (0.85396,1.3593) (0.85396,1.36357) (0.85924,1.36357) (0.85924,1.35664) (0.86146,1.35664) (0.86146,1.35793) (0.86155,1.35793) (0.86155,1.36231) (0.86193,1.36231) (0.86193,1.35608) (0.86251,1.35608) (0.86251,1.35668) (0.86716,1.35668) (0.86716,1.3663) (0.87087,1.3663) (0.87087,1.37086) (0.87462,1.37086) (0.87462,1.46071) (0.87686,1.46071) (0.87686,1.46319) (0.87864,1.46319) (0.87864,1.47458) (0.87942,1.47458) (0.87942,1.44521) (0.88052,1.44521) (0.88052,1.45132) (0.88572,1.45132) (0.88572,1.46458) (0.8921,1.46458) (0.8921,1.47088) (0.89532,1.47088) (0.89532,1.54921) (0.9003,1.54921) (0.9003,1.55255) (0.91111,1.55255) (0.91111,1.55972) (0.91188,1.55972) (0.91188,1.56696) (0.91331,1.56696) (0.91331,1.65745) (0.92376,1.65745) (0.92376,1.66573) (0.92491,1.66573) (0.92491,1.6837) (0.93117,1.6837) (0.93117,1.73304) (0.93521,1.73304) (0.93521,1.78565) (0.95191,1.78565) (0.95191,1.90731) (0.96293,1.90731) (0.96293,1.97235) (0.97359,1.97235) (0.97359,2.12457) (1.0,2.12457)};
\addplot[blue, dashed, line width=0.9pt, domain=0:1, samples=2] {1.4729174};
\addplot[red, line width=0.9pt, domain=0:1, samples=2] {1};
\addplot[only marks, mark=*, mark size=1.6pt, black, forget plot] coordinates {(0.75398,1.22881)};
\end{groupplot}
\end{tikzpicture}
\caption{
Long-term mean cost $LTMC^\gamma$ of the MTTF-threshold policies, relative to the cost $v^\ast_{\mathrm{state}}$ of the optimal state-based policy (solid red line), as a function of the threshold $\gamma$, for two configurations and cost pairs; the dot marks the best threshold, and the dashed blue line the cost of the best $r$-out-of-$n$ policy. The curves are piecewise constant, with 54 and 731 breakpoints, and they are not unimodal, i.e., they increase and then decrease again as $\gamma$ varies. In (a), the best MTTF-threshold policy is within $0.24\%$ of the optimal state-based policy, while the best $r$-out-of-$n$ policy is $6.2\%$ above it; (b) is one of the cases with the largest gaps, where the best MTTF-threshold policy is $23\%$ and the best $r$-out-of-$n$ policy $47\%$ above the optimum. In (a), corrective maintenance ($\gamma \to 0$) costs $20$ times the optimum, and in (b) $5.4$ times.
}
\label{fig:curves}
\end{figure}
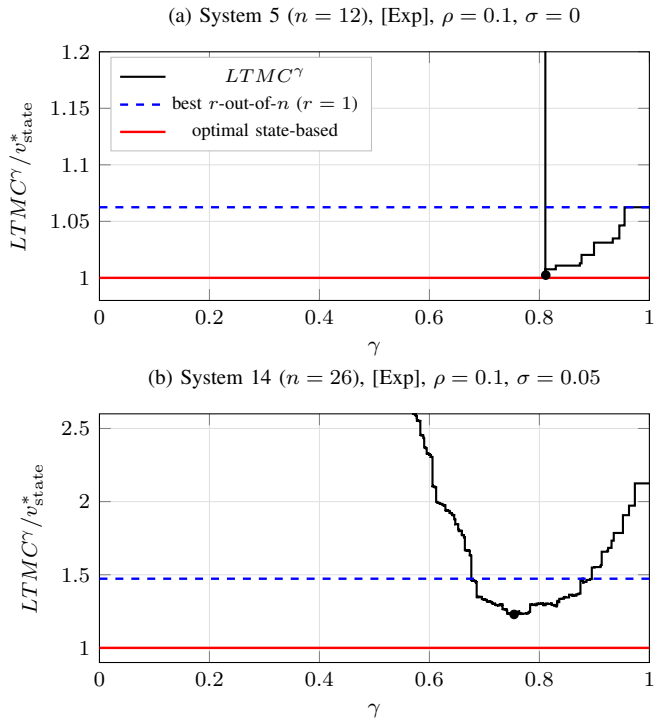

\subsection{Results: comparison with the optimal policies}\label{subsec:results optimal}
Table~\ref{tab:gaps} and Figure~\ref{fig:gaps} compare the best MTTF-threshold policy with the optimal state-based policy, through the relative excess cost $v^\ast_\gamma / v^\ast_{\mathrm{state}} - 1$, over the $26{,}505$ combinations with $n \leq 27$.

\paragraph{The best MTTF-threshold policy is near-optimal}
It is optimal, up to a relative error of $10^{-9}$, in $17\%$ of the combinations, and its excess cost has a median of $0.36\%$, a 90th percentile of $3.4\%$ and a 99th percentile of $13\%$.
By Section~\ref{sec:dominance}, this also bounds what the signature-threshold policies can achieve: in half of the combinations no weight vector can improve on the best MTTF threshold by more than $0.36\%$.

The optimal state-based policy genuinely uses the identity of the failed components: in $95\%$ of the combinations there is at least one number of failed components at which it repairs at some, but not all, of the working states, and the median number of such levels is 6.

\paragraph{When does the MTTF threshold fall short?}
The gap of the best MTTF-threshold policy decreases with the pairwise beta factor, from a median of $1.4\%$ for [Exp] to $0.04\%$ in the stressed scenario: when failures come in large batches, the next failure epoch is likely to bring the system down whatever its state, so the information on the state, and its summary by the MTTF, matter less. Its median stays below $0.5\%$ for every cost ratio, while its tail shrinks as $\rho$ grows, with a 90th percentile from $5.9\%$ at $\rho = 0.1$ to $1.1\%$ at $\rho = 10$ (Figure~\ref{fig:gaps}(a)); it is also largest for small positive setup costs, with a 90th percentile of $5.5\%$ at $\sigma = 0.05$ and $1.4\%$ at $\sigma = 0.5$. It tends to grow with the size of the system, from zero for $n = 5$ to a median of $1.7\%$ and a 90th percentile of $8.9\%$ for $n = 27$ (Figure~\ref{fig:gaps}(b)), with the notable exception of the grid network 20 ($n = 24$), where the median gap is $0.03\%$.
The largest gaps, of about $30\%$ to $35\%$, occur for $\rho = 0.1$, $\sigma = 0.05$, nearly independent failures ([Exp], or $\beta = 1\%$) and the largest systems; for instance, for system 17 under [Exp] the best MTTF-threshold policy, with $\gamma = 0.78$, costs $35\%$ more than the optimal state-based policy, while the best $r$-out-of-$n$ policy, with $r = 4$, costs at least $80\%$ more. With cheap components and an expensive system failure, the optimal policy can afford to repair at states that are close to failure in a way that a single number does not capture, such as states where one further failure brings the system down; by Section~\ref{sec:dominance}, these are the cases where weighting the per-state signature differently could pay off.

\begin{table}[!t]

\centering
\caption{
Relative excess cost (\%) of the best MTTF-threshold policy over the optimal state-based policy, $v^\ast_\gamma / v^\ast_{\mathrm{state}} - 1$: fraction of combinations where the excess is at most $10^{-9}$ (``opt.''), median, 90th percentile and maximum, over the 19 systems with $n \leq 27$, the models with the given pairwise beta factor $\beta$ (only [Exp] for $\beta = 0$) and the 45 cost pairs.
}
\label{tab:gaps}
\small
\setlength{\tabcolsep}{2.6pt}
\begin{tabular}{@{}lrrrr@{}}
\toprule
$\beta$ & opt. & med. & p90 & max \\
\midrule
$0$ ([Exp]) & $8.7$ & $1.37$ & $7.0$ & $34.8$ \\
$1\%$ & $12.8$ & $1.07$ & $6.4$ & $33.5$ \\
$5\%$ & $13.7$ & $0.67$ & $4.3$ & $29.0$ \\
$10\%$ & $15.9$ & $0.43$ & $2.7$ & $24.4$ \\
$20\%$ & $18.1$ & $0.19$ & $1.5$ & $17.5$ \\
$50\%$ & $24.6$ & $0.04$ & $0.6$ & $6.8$ \\
\midrule
All & $16.7$ & $0.36$ & $3.4$ & $34.8$ \\
\bottomrule
\end{tabular}
\end{table}

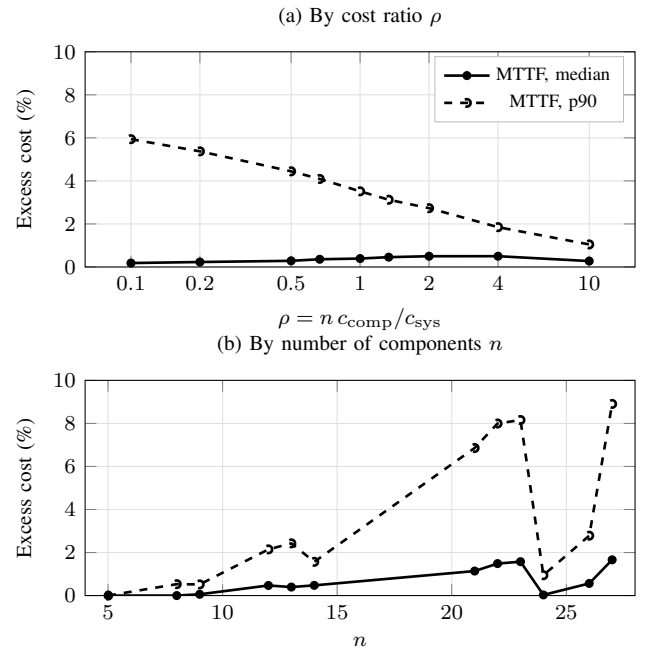
\begin{figure}[!t]
\centering
\begin{tikzpicture}
\begin{groupplot}[
  group style={group size=1 by 2, vertical sep=1.5cm},
  width=\linewidth, height=0.5\linewidth,
  ymin=0, ylabel={Excess cost (\%)},
  grid=both, grid style={line width=0.3pt, draw=gray!25},
  tick label style={font=\footnotesize}, label style={font=\footnotesize},
  title style={font=\footnotesize},
  legend style={font=\scriptsize, draw=gray!50, at={(0.98,0.98)}, anchor=north east, legend columns=1},
]
\nextgroupplot[title={(a) By cost ratio $\rho$}, xmode=log, xlabel={$\rho = n\,\ccomp/\csys$},
  xtick={0.1,0.2,0.5,1,2,4,10}, xticklabels={$0.1$,$0.2$,$0.5$,$1$,$2$,$4$,$10$}, ymax=10]
\addplot[black, line width=1pt, mark=*, mark size=1.2pt] coordinates
  {(0.1,0.183) (0.2,0.231) (0.5,0.286) (0.6667,0.359) (1,0.392) (1.3333,0.457) (2,0.500) (4,0.501) (10,0.275)};
\addlegendentry{MTTF, median}
\addplot[black, line width=1pt, dashed, mark=o, mark size=1.4pt] coordinates
  {(0.1,5.940) (0.2,5.374) (0.5,4.446) (0.6667,4.091) (1,3.517) (1.3333,3.125) (2,2.734) (4,1.851) (10,1.048)};
\addlegendentry{MTTF, p90}
\nextgroupplot[title={(b) By number of components $n$}, xlabel={$n$},
  xtick={5,10,15,20,25}, xmin=4, xmax=28, ymax=10]
\addplot[black, line width=1pt, mark=*, mark size=1.2pt] coordinates
  {(5,0.000) (8,0.000) (9,0.059) (12,0.469) (13,0.395) (14,0.474) (21,1.139) (22,1.486) (23,1.575) (24,0.030) (26,0.563) (27,1.661)};
\addplot[black, line width=1pt, dashed, mark=o, mark size=1.4pt] coordinates
  {(5,0.000) (8,0.530) (9,0.518) (12,2.146) (13,2.426) (14,1.568) (21,6.857) (22,7.995) (23,8.158) (24,0.946) (26,2.783) (27,8.906)};
\end{groupplot}
\end{tikzpicture}
\caption{
Relative excess cost of the best MTTF-threshold policy over the optimal state-based policy: median (solid) and 90th percentile (dashed) over the 19 systems with $n \leq 27$, all models and the cost pairs, (a) for each cost ratio $\rho$, and (b) for each system size $n$ (several systems share the sizes $8$, $12$, $13$, $14$ and $21$).
}
\label{fig:gaps}
\end{figure}

\subsection{Computational effort of the policies}\label{subsec:runtimes}
Once the classes are stored, the whole curve $\gamma \mapsto (T^\gamma, N^\gamma, P^\gamma)$ of a failure model takes at most $2.8$ seconds (Table~\ref{tab:effort}), and the best threshold for a new cost pair is then found by evaluating~\eqref{eq:LTMC gamma} at every breakpoint, without recomputing anything else. The best $r$-out-of-$n$ policy takes a fraction of a second.
The optimal state-based policy, instead, needs a pass over the $2^n$ states for each iteration of Dinkelbach's method and for each cost pair. Dinkelbach's method converged in all $26{,}505$ runs, with $2.4$ passes on average and never more than five. One pass takes $0.8$ seconds for $n = 21$, $3.0$ seconds for $n = 23$ and $52$ seconds for $n = 27$, so the optimal state-based policy for one cost pair takes a median of $197$ seconds for system 17, and about $2.5$ hours for the 45 cost pairs of one configuration, with a peak memory of $9.3$~GB; extrapolating the $n^2 2^n$ growth, it would take about a day per configuration and about 75~GB of memory for $n = 30$, which is why we did not compute it for the three largest systems.
This is the practical case for the MTTF-threshold policies: they come close to the optimal state-based policy at a small fraction of its cost, their cost curve serves every cost pair at once, and they are specified by a single number, the threshold $\gamma$, applied to an index, the MTTF, that is meaningful to practitioners.

\section{Conclusions and discussion}\label{sec:conclusions}
We presented an algorithm that computes the structural signature and the MTTF of a monotone binary system at each of its $2^n$ states, with the same $2^n$ evaluations of the structure function needed for the signature of the complete system alone and a polynomial overhead in operations and memory, for exchangeable failure times whose survivor count is a continuous-time Markov chain, a setting that admits simultaneous failures. Under the L\'evy-frailty Marshall-Olkin distribution, the monotonicity of the MTTF yields the long-run cost of every MTTF-threshold policy in a single pass over classes of states, and the per-state signature is a strictly finer decision criterion than the MTTF.

On 22 benchmark networks with up to 30 components, the per-state computation took, for $n = 30$, between $2.0$ and $4.4$ times as long as the enumeration for the signature of the complete system, and the best MTTF-threshold policy, tuned over all breakpoints of its cost curve, which is not unimodal, i.e., can increase and then decrease again as the threshold varies, was within a median of $0.36\%$ of the optimal state-based policy. This supports a simple rule for practitioners: monitor which components have failed, and repair when the MTTF of the current state falls below a fraction $\gamma$ of that of a new system, with $\gamma$ tuned exactly by our method. The rule is transparent and computed in seconds, whereas the optimal state-based policy costs $\mathcal O(n^2 2^n)$ operations per iteration and per cost setting.

Our study has limitations that suggest further work.
The algorithm enumerates the $2^n$ states, which limits it to a few dozen components; combining Lemma~\ref{lemma:Wk} with binary decision diagrams~\cite{kuo2007efficient} may extend it to larger systems.
The single-pass evaluation of the threshold policies relies on the monotonicity of the MTTF, which can fail for other models satisfying our assumptions, such as sequential order statistics in which the failure rate of the working components decreases as other components fail.
Non-identical components, repair delays as in~\cite{lagos2026delayed}, imperfect repairs and other families of systems are natural extensions.
Finally, the gap between the best MTTF threshold and the optimal state-based policy, largest for nearly independent failures, cheap components and large systems, bounds what weighting the per-state signature differently can achieve; finding weights that close it, and comparing states with different numbers of failed components, remain open.

\appendices
\section{Proofs}\label{app:proofs}

\begin{proof}[Proof of Lemma~\ref{lemma:Wk}]
Consider a state $\mathbf x$ in $\{0,1\}^n$ and fix $k=1,\ldots, |\mathbf x|$. Consider now a sub-state $\mathbf y$ of $\mathbf x$ with $k$ less working components, i.e., $\mathbf y$ satisfies $\mathbf y \leq \mathbf x$ and $|\mathbf y|=|\mathbf x|-k$.
Since $\mathbf y$ was obtained shutting down $k$ working components of $\mathbf x$, by switching back on $k-1$ of them, we obtain a state in $\mathrm{Sons}(\mathbf x)$.
There are exactly $k$ different ways of  doing the latter, one per component that we leave shut-off, so there are exactly $k$ different states in $\mathrm{Sons}(\mathbf x)$, say $\mathbf x^1$, \ldots, $\mathbf x^k$, for which after shutting down $k-1$ of their working components we obtain $\mathbf y$.
It follows that if $\Phi(\mathbf y)=1$ then it will be counted once in the counting of $W^{\mathbf x}_k$, and also once in the counting of all $W^{\mathbf x^1}_{k-1}$, \ldots, $W^{\mathbf x^k}_{k-1}$.
This proves the result.
\end{proof}

\begin{proof}[Proof of Theorem~\ref{theo}]
1) In Lemma~\ref{lemma:Wk}, $\Phi$ enters only through $W^{\mathbf x}_0 = \Phi(\mathbf x)$, once per state.

2) By~\eqref{def:bar S} and~\eqref{eq:E Samaniego}, it suffices to compute $W_k = \sum_{\mathbf x : |\mathbf x| = n-k} \Phi(\mathbf x)$ for $k = 0, \ldots, n$, which takes $\sum_{k=0}^n \binom n k = 2^n$ additions of bits, and then~\eqref{eq:E Samaniego} takes $\mathcal O(n)$ operations. Processing one level at a time requires at most $\binom n {\lfloor n/2 \rfloor}$ bits, and $\binom n {\lfloor n/2 \rfloor} \sim 2^n / \sqrt{\pi n/2}$ as $n \to \infty$.

3) A state with $k$ failures has $n-k$ sons, each with a counting vector of at most $n-k+1$ entries, so~\eqref{eq:Wk} takes $\mathcal O((n-k)^2)$ integer additions and memory accesses and $n-k$ divisions, and~\eqref{eq:E Samaniego sub} takes $\mathcal O(n-k)$ operations. Since $\sum_{k=0}^n \binom n k (n-k) = n 2^{n-1}$ and $\sum_{k=0}^n \binom n k (n-k)^2 = n(n+1) 2^{n-2}$, summing over the $\binom n k$ states of each level gives the operation counts. Algorithm~\ref{alg} stores one MTTF per state and the counting vectors of two consecutive levels, i.e., at most $2(n+1)\binom n {\lfloor n/2 \rfloor} = \mathcal O(\sqrt n \, 2^n)$ integers.
\end{proof}

\begin{proof}[Proof of Proposition~\ref{prop:all thresholds}]
If $\mathbf x \in \mathcal C_\gamma$ and $\mathbf y \geq \mathbf x$, then $\Phi(\mathbf y) = 1$ and $m(\mathbf y) \geq m(\mathbf x) \geq \gamma \, m(\mathbf 1)$, so $\mathbf y \in \mathcal C_\gamma$. Since components are not repaired within a cycle, every state visited before a state $\mathbf x \in \mathcal C_\gamma$ is larger than $\mathbf x$, and hence also in $\mathcal C_\gamma$. Therefore, up to the time it reaches $\mathbf x$, the system under $\pi^\gamma$ evolves as the system without repairs, and the expected time it spends in $\mathbf x$ during a cycle is $\tau_{|\mathbf x|}$, since the arrival at $\mathbf x$ is the same event in both systems and the holding time in $\mathbf x$ has mean $1/\nu_{|\mathbf x|}$. The cycle ends at the first exit from $\mathcal C_\gamma$, so $T^\gamma$ is the total expected time spent in $\mathcal C_\gamma$. The number of failed components at the repair equals the number of component failures during the cycle, all of which occur at jumps out of states of $\mathcal C_\gamma$; the expected number of failures from $\mathbf x$ is the expected time spent in $\mathbf x$ multiplied by the rate $\varphi_{|\mathbf x|}$, which gives $N^\gamma$. Similarly, the repair is corrective iff the cycle ends with a jump from a state of $\mathcal C_\gamma$ to a failed state; from $\mathbf x$, a jump to one of the $\binom{|\mathbf x|}{|\mathbf x| - w'}$ sub-states with $w'$ working components occurs at rate $\nu_{|\mathbf x|} P_{|\mathbf x|, w'}$, and the target is chosen uniformly among them, of which a fraction $1 - W^{\mathbf x}_{|\mathbf x| - w'} / \binom{|\mathbf x|}{|\mathbf x| - w'}$ are failed, which gives $P^\gamma$. Finally, every cycle ends with exactly one repair, at the state $\mathbf y$ where the cycle ends, and its cost~\eqref{eq:cost setup} is $\repcost(\mathbf y) = \cset + (n - |\mathbf y|) \ccomp + \mathbf 1_{\{\Phi(\mathbf y) = 0\}} \csys$, where $n - |\mathbf y|$ is the number of component failures during the cycle and $\mathbf 1_{\{\Phi(\mathbf y) = 0\}}$ indicates that the repair is corrective; taking expectations gives $c^{\pi^\gamma}_{\mathrm{rep}}(\mathbf 1)$.
\end{proof}

\begin{proof}[Proof of Lemma~\ref{lemma:monotone}]
Write $A$ and $A' \subseteq A$ for the sets of components working in $\mathbf x$ and $\mathbf x'$, respectively. By the memoryless property of the L\'evy-frailty Marshall-Olkin distribution~\cite{mai2009levy}, the residual lifetimes $(R_i)_{i \in A}$ of the components working at the moment the system arrives at $\mathbf x$ follow a L\'evy-frailty Marshall-Olkin distribution of dimension $|A|$ with the same Laplace exponent $\Psi$, and the sub-vector $(R_i)_{i \in A'}$ follows the one of dimension $|A'|$, which is the law of the residual lifetimes when the system arrives at $\mathbf x'$. Using the same vector $(R_i)_{i \in A}$ for both, the state at time $t$ started from $\mathbf x'$ is $\mathbf y'_t := (\mathbf 1_{\{i \in A', \, R_i > t\}})_i$, which is componentwise smaller than the state $\mathbf y_t := (\mathbf 1_{\{i \in A, \, R_i > t\}})_i$ started from $\mathbf x$. By monotonicity of $\Phi$, the system fails no later from $\mathbf x'$ than from $\mathbf x$ on every sample path, and taking expectations gives $m(\mathbf x') \leq m(\mathbf x)$.
\end{proof}

\begin{proof}[Proof of Theorem~\ref{theo:dominance}]
By~\eqref{eq:m as weighted sum}, $m(\mathbf x) - m(\mathbf y) = \sum_{k=0}^{w-1} ( \bar S^{\mathbf x}_k - \bar S^{\mathbf y}_k ) \, \delta_k^{w}$, a sum of nonnegative terms, which is positive if some term has both factors positive.
\end{proof}

\section*{Acknowledgment}

The authors acknowledge the computational resources provided by the High-Performance Computing infrastructure of the Faculty of Engineering and Sciences at Universidad Adolfo Ib\'a\~nez, which were instrumental in generating the core results presented in this work.
The authors used the AI system Claude (Anthropic) in the preparation of this article, as follows: to draft parts of the text, mainly in Sections~\ref{sec:related}, \ref{sec:policies} and~\ref{sec:experiments}, which the authors then revised; to produce the TikZ code of Figures~\ref{fig:simultaneity one}--\ref{fig:gaps}; and to write the code of the computational experiments of Section~\ref{sec:experiments}, which the authors checked against independent implementations. The authors reviewed all AI-generated content and take full responsibility for the article.

\ifCLASSOPTIONcaptionsoff
  \newpage
\fi

\bibliographystyle{IEEEtran}
\bibliography{IEEEabrv,bibliography}

\end{document}